\documentclass[a4paper,11pt]{article}
\usepackage{amsmath,amsfonts,amssymb}
\usepackage{graphicx}

\usepackage{fullpage}
\newtheorem{theorem}{Theorem}
\newtheorem{lemma}{Lemma}
\newtheorem{result}{Result}
\newtheorem{definition}{Definition}

\newenvironment{proof}{\noindent {\bf Proof:\,\ }}{\hfill\mbox{$\Box$}\smallskip}

\usepackage{color}

\newtheorem{observation}{Observation}
\newtheorem{remark}{Remark}

\newtheorem{corollary}{Corollary}

\usepackage{wrapfig}
\usepackage{comment}
\usepackage{graphicx}
\usepackage{xcolor}
\usepackage{amsmath}
\usepackage{gensymb}
\usepackage{lineno}
\usepackage{todonotes}

\newcommand{\remove}[1]{}

\usepackage{color}

\title{Partial Domination in Some Geometric Intersection Graphs and Some Complexity Results\footnote{A preliminary version of this paper appeared in the proceedings of CALDAM 2025.}}

\author{Madhura Dutta$^{1,2}$\footnote{Supported by DST INSPIRE PhD Fellowship} \and 
Anil Maheshwari$^3$\footnote{Research supported by NSERC} \and  
Subhas C. Nandy$^4$ \and Bodhayan Roy$^5$}

\date{$^1$TCG CREST, Kolkata-700091, India; {\tt madhura1998kkg@gmail.com}\\
$^2$Academy of Scientific and Innovative Research (AcSIR), Ghaziabad-201002, India \\
$^3$School of Computer Science, Carleton University, Ottawa, Canada; {\tt anil@scs.carleton.ca} \\
$^4$Ramkrishna Mission Vivekananda Centenary College, Kolkata-700118, India; {\tt subhas.c.nandy@gmail.com} \\
$^5$Indian Institute of Technology Kharagpur, West Bengal-721302, India; {\tt bodhayan.roy@gmail.com}}

\begin{document}
\maketitle

\begin{abstract}
{\em Partial domination problem} is a generalization of the {\em minimum dominating set problem} on graphs. Here, instead of dominating all the nodes, one asks to dominate at least a fraction of the nodes of the given graph by choosing a minimum number of nodes. For any real number $\alpha\in(0,1]$, $\alpha$-partial domination problem can be proved to be NP-complete for general graphs. In this paper, we define the {\em maximum dominating $k$-set} of a graph, which is polynomially transformable to the partial domination problem. The existence of a graph class for which the minimum dominating set problem is polynomial-time solvable, whereas the partial dominating set problem is NP-hard, is shown. We also propose polynomial-time algorithms for the maximum dominating $k$-set problem for the unit and arbitrary interval graphs. The problem can also be solved in polynomial time for the intersection graphs of a set of 2D objects intersected by a straight line, where each object is an axis-parallel unit square, as well as in the case where each object is a unit disk. Our technique also works for axis-parallel unit-height rectangle intersection graphs, where a straight line intersects all the rectangles. Finally, a parametrized algorithm for the maximum dominating $k$-set problem in a disk graph where the input disks are intersected by a straight line is proposed; here the parameter is the ratio of the diameters of the largest and smallest input disks.
\end{abstract}


\section{Introduction}

The {\em minimum dominating set} problem in a simple undirected graph $G=(V,E)$ is well-studied in graph theory. The objective is to choose a set $S \subseteq V$ of minimum cardinality such that for every node $v \in V\setminus S$, there exists at least one member $v' \in S$ satisfying $v \in N(v')$, where $N(v')$ is the set of nodes adjacent to $v'$ in the graph $G$. The {\em minimum dominating set} problem is an NP-complete problem \cite{garay}. The hardness status remains unchanged for bipartite graphs and split graphs \cite{bertossi}. However, a $(1+\log |V|)$-factor approximation algorithm exists as the dominating set problem is a special instance of the set cover problem \cite{johnson}. It admits a PTAS for planar graphs \cite{baker} and unit disk graphs \cite{hunt}. The problem is linear-time solvable for trees using dynamic programming \cite{ref2, 12} and for interval graphs, assuming the endpoints of the intervals corresponding to the nodes are in sorted order \cite{refx}. For the class of geometric intersection graphs of axis-parallel rectangles intersecting a straight line and some of its sub-classes, the domination problem has been studied in \cite{new1}. For a comprehensive survey on the dominating set problem, see \cite{b,a}.

Recently, in the context of communication networks, the problem of partial domination has become important. Usually, it is difficult to cover the entire network with limited resources. So, dominating (covering) a large subset using limited resources becomes important. In partial domination of a graph $G=(V,E)$, a real number $\alpha \in (0,1]$ is given; the objective is to find a set $S\subseteq V$ of minimum cardinality such that the size of the closed neighborhood of the set $S$, denoted by $N[S]=S \cup N(S)$, is greater than or equal to $\alpha$ times the size of the set $V$ (i.e., $|N[S]| \geq \alpha|V|$) \cite{ref1,ref3}, where $N(S)=\{u| ~ u \in V\setminus S,  ~\text{and} ~(u,v) \in E ~\text{for some}~ v\in S\}$. For $\alpha=1$, $\alpha$-partial dominating set is the dominating set of the graph. The theoretical study of partial domination is rarely observed in the literature. In \cite{ref1}, methods for evaluating the $\frac{1}{2}$-partial dominating set for some graph classes, like cycles, paths, grids, etc., are suggested. Some practical applications of partial domination are also available in  \cite{ref1}. A parameterized algorithm for the partial dominating set problem in arbitrary graphs is given in \cite{ref1}, using the results of \cite{kneis}. Sub-exponential algorithms for this problem for planar graphs and apex-minor-free graphs are proposed in \cite{new2} and \cite{new3}, respectively.

Several variants of the dominating set problem have been studied \cite{das,domke,hedet,lan,ruba}. Here, we introduce the {\em maximum dominating $k$-set} problem for a graph. The objective is to dominate a maximum number of nodes in the graph with a suitably chosen subset of $k$ nodes, where $k \geq 0$ is an integer. Naturally, the maximum dominating $k$-set problem has many applications where a limited number of facilities are available, and the goal is to use them to cover as many locations as possible. This is similar to the maximum coverage problem, a variant of the {\em set cover problem}; here, we must choose at most k sets whose union has the maximum cardinality. The notion of a max $k$-vertex cover, where the goal is to cover the maximum number of edges in a graph using $ k$ vertices, is also analogous. Parametrized algorithms for the maximum $k$-set cover have been studied in \cite{jn13}.

Another similar concept is that of a budgeted dominating set, where the goal is to maximize the number of dominated vertices using a vertex subset of size less than or equal to the given budget \cite{jn7,jn8}. This has been studied in uncertain graphs; some versions of budgeted domination, namely, connected and partial-connected versions, have also been studied \cite{jn10,jn11,jn9,jn12}. For unit disk graphs, parallel algorithms for minimum general partial dominating set and maximum budgeted dominating set have been given in \cite{jn14}.

In this connection, it is worth noting that studying partial optimal solutions to specific optimization problems is an active area of research. In social networks, the maximum $k$-cover problem was studied in \cite{KKD}, where the objective was to choose an influential subset of $k$ individuals to reach (cover) as many individuals as possible. They show that the problem is NP-hard and propose an approximation algorithm. In the context of geometric optimization, an $O(kn^2)$ time and $O(kn)$ space algorithm is available for the maximum hitting $k$-set problem for unit intervals where the objective is to place $k$ points that hit the maximum number of unit intervals among $n$ unit intervals that are arranged on a real line  \cite{ahn}. In \cite{De}, modified {\em local search} based PTASes have been proposed for computing the maximum hitting $k$-set problem for the intersection graphs of homothetic copies of convex objects, which includes disks, squares of arbitrary sizes, regular $m$-gons, translated and scaled copies of a convex object, etc. 

\subsection{Preliminaries}
\begin{definition}
Given a graph $G=(V,E)$ and a real number $\alpha \in (0,1]$, an {\em $\alpha$-partial dominating set} is a set $S\subseteq V$ such that $|N[S]|\geq \alpha\cdot |V|$. The cardinality of an $\alpha$-partial dominating set of minimum size is called the {\em $\alpha$-partial domination number} of the graph, and is denoted by $\gamma_{\alpha}$, i.e., 
$\gamma_{\alpha}=\min\{|S|:S\subseteq V \text{ with }|N[S]|\geq \alpha\cdot |V| \}$.
\end{definition}

\begin{definition}\label{defkset}
Given a graph $G=(V,E)$ and an integer $k$ ($\leq |V|$), a {\em maximum dominating $k$-set} is a set $S\subseteq V$ with $|S|=k$ and $|N[S]|$ is maximized, i.e., for all $S'\subseteq V$ with $|S'|=k$, we have $|N[S]|\geq |N[S']|$. Here,  $|N[S]|$  is referred to as the {\em maximum dominated neighborhood size of a $k$-set}.
\end{definition}

\begin{definition}
A {\em geometric intersection graph} $G=(V,E)$ is a graph where each vertex $v_i\in V$ corresponds to a geometric object, and an edge $(v_i,v_j) \in E$ if and only if the objects corresponding to the vertices $v_i$ and $v_j$ have non-empty intersection. 
\end{definition}
For example, an interval graph is a geometric intersection graph of a set of intervals arranged on a straight line, where two vertices are adjacent if and only if the corresponding intervals overlap. Many hard problems in graph theory can be solved efficiently for the geometric intersection graph when the geometric objects corresponding to that graph satisfy specific geometric properties.

\subsection{New results} 
In this paper, we study the maximum dominating $k$-set problem on graphs, where the objective is to dominate the maximum number of nodes in the graph using a suitably chosen subset of $k$ nodes. We present the following results.
\begin{itemize}  
\item We show that the maximum dominating $k$-set problem and the partial domination problem are polynomially equivalent. 
\item The decision version of the partial dominating set problem (for a given $\alpha\in (0,1]$) is shown to be NP-hard for arbitrary graphs using a polynomial time reduction from the dominating set problem. 
\item We show the existence of a graph class for which domination is polynomial time solvable, but partial domination is NP-hard. 
\item We provide polynomial-time algorithms for the maximum dominating $k$-set problem for different types of geometric intersection graphs. 
\begin{itemize}
\item For unit interval graphs, our proposed algorithm runs in $O(nk\log n)$ time using $O(n)$ space. Using this, we propose an algorithm for arbitrary interval graphs that runs in $O(n^2k)$ time. \\
\item We propose an algorithm for the maximum dominating $k$-set problem for a unit square intersection graph, where the input squares are axis-parallel and intersected by a straight line; the time complexity of this algorithm is a polynomial function of the number of input squares. This method is extended to solve the same problem for unit-height rectangles intersected by a line (provided the line meets certain criteria) and unit disks intersected by a line in polynomial time. \\
\item We present an algorithm for disk graphs when a line intersects all the disks. The time complexity of this algorithm depends on the ratio of the diameters of the largest and smallest disks.
\end{itemize}
\end{itemize}

Our interest in studying the partial dominating set problem for certain restricted classes of geometric intersection graphs stems from the fact that polynomial-time algorithms exist for the dominating set problem in these graph classes (see \cite{refx}, \cite{new1}). In \cite{new3,new2}, subexponential algorithms for the partial domination problem are proposed for planar graphs and apex minor-free graphs. For graphs with bounded local treewidth, these algorithms run in polynomial time. Note that interval graphs do not have bounded local treewidth, need not be planar, and can contain apex graphs as minor. The same applies to intersection graphs of $ d$-dimensional ($d\geq 2$) geometric objects. This paper focuses on the partial domination problem for interval graphs, the intersection graph of unit squares intersected by a straight line, the intersection graph of unit height rectangles intersected by a straight line, the unit disk graph where a straight line intersects all unit disks, and the disk graph where a straight line intersects all the disks.

\subsection{Organization}
In the next section, we prove that the partial domination problem is NP-hard for arbitrary graphs and show that it is equivalent to the maximum dominating $k$-set problem in the context of computational complexity. Then, the existence of the graph class for which the domination problem is in complexity class P, whereas partial domination is NP-complete, is shown. The maximum dominating $k$-set problem for unit interval graphs is studied in Section \ref{interval}. We extend our technique to work for the intersection graph of arbitrary-sized intervals in Section \ref{arbint}. In Sections \ref{square} and \ref{rectangle}, we explore the problem for the unit square intersection graph and the unit-height rectangle intersection graph, respectively, where a line intersects all the geometric objects (squares and rectangles, respectively) corresponding to the graph nodes. 
In Section \ref{di}, we study the problem for the unit disk graph and the disk graph, respectively, where a straight line intersects all the input unit disks and disks. Finally, Section \ref{conclusion} gives some concluding remarks.

\section{Complexity results for the partial domination problem}
For a given graph $G=(V,E)$, a real number $\alpha \in (0,1]$ and an integer $\kappa\leq |V|$, the partial domination problem asks whether there exists an $\alpha$-partial dominating set of $G$ of size $\leq \kappa$.

\subsection{NP-hardness}
\begin{theorem} \label{t1}
The decision version of the partial dominating set problem is NP-complete.
\end{theorem}

\begin{proof}
The problem is in NP, since given a subset $S \subseteq V$, we can count the number of nodes in $V$ that are dominated by members in $S$ in polynomial time. We now prove the NP-hardness by providing a polynomial-time reduction from the dominating set problem of a graph. 

Given a graph $G=(V,E)$, an integer $\kappa$, and a real number $\alpha \in (0,1]$, we construct a graph $G'=(V',E')=(V_1\cup V_2, E')$, where $V_1=V$ and $V_2$ consists of $(\lfloor n/\alpha\rfloor - n)$ isolated nodes, where $n=|V|$. Note that $|V'|=|V_1\cup V_2|=|V_1|+|V_2|=\lfloor n/\alpha\rfloor$. The edge set $E'$ consists of the edges present in $G$ only, i.e., $E'=E$.

We first show that the minimum dominating set for the graph $G$ can be obtained by executing the algorithm for the partial domination problem on the graph $G'$ with parameters $\alpha$ and $\kappa$. We run the algorithm for the partial domination on the graph $G'$ with parameters $\alpha$ and $\kappa$. Let it return a set $S'$ with $|S'|\leq \kappa$, and  $|N[S']|\geq\alpha\cdot|V'|=\alpha\cdot\lfloor n/\alpha\rfloor\geq n$ (since $|N[S']|$ is an integer). Now, consider the following cases:

\begin{description}
\item[Case-1 - $S'\subseteq V_1$:] Here,  $S'$ itself becomes a dominating set of $G$.

\item[Case-2 - $S'=S_1'\cup S_2'$, where $S_1' \subseteq V_1$ and $S_2'\subseteq V_2$:] Since $S_1'\cap S_2'=\emptyset$ and $N[S_1']\cap N[S_2']=\emptyset$, we have $|N[S']|=|N[S_1'\cup S_2']|=|N[S_1']|+|N[S_2']|=|N[S_1']|+|S_2'|$.
We already have $|N[S']|\geq n \implies |N[S_1']|\geq n-|S_2'|$. Therefore, at most $|S_2'|$ nodes of $G$ are not dominated by $S_1'$ as $N[S_1']\cap N[S_2']=\emptyset$.
We now take all the nodes in $S_3=V \setminus N[S_1']$, and report $S_1'\cup S_3$ as the dominating set for the graph $G$. As the size of $S_3= V\setminus N[S_1']$ is at most $|S_2'|$, the size of the reported dominating set is at most $|S'|$.

\item[Case-3 - $S'\subseteq V_2$:] All the nodes in $V_2$ are isolated and since $|N[S']|\geq n$, we have $|S'|\geq n$. So, we can take all the members of $V$ instead of $S'$, which will be an $\alpha$-partial dominating set for the graph $G'$, as well as a dominating set of $G$ of size at most $|S'|$.
\end{description}
Conversely, if $D$ is a dominating set of $G$, then $D$ would be an $\alpha$-partial dominating set of $G'$ by construction.  
\end{proof}

The dominating set problem is a well-known NP-hard problem for general graphs. Thus, Theorem \ref{t1} says that the partial domination problem 
is NP-complete for those graph classes where the dominating set problem 
is NP-hard.

\subsection{Computational equivalence of  partial domination problem and maximum dominating $k$-set problem}

\begin{theorem} \label{t2}
Maximum dominating $k$-set problem and $\alpha$-partial dominating set problem are computationally equivalent for any graph $G=(V,E)$. 
\end{theorem}

\begin{proof}
Let us first assume that an algorithm for the partial domination problem is known, and for a given integer $k\in \{1, \ldots, n\}$, we want to solve the max dominating $k$-set problem.

Let $|V|=n$. Using the said algorithm for $\alpha_i=i/n$, for $i = 1, 2, \dots, n$, we compute the $\alpha_i$-partial domination number $\gamma_{\alpha_i}$. Note that, $\gamma_{\alpha_1}\leq \gamma_{\alpha_2} \leq \dots \leq \gamma_{\alpha_n}=\gamma$, where $\gamma$ is the domination number of the graph $G$. Our objective is to solve the max dominating $k$-set problem. 

If $k\geq \gamma_{\alpha_n}(=\gamma)$, then any $k$-size subset of $V$ containing the $\alpha_n$-partial dominating set is a max dominating $k$-set of $G$. If  $k \in [\gamma_{\alpha_i},\gamma_{\alpha_{i+1}})$, where $i+1 \leq n$, we have maximum dominated neighborhood size of a $k$-set is $i$, and any $k$-size subset of $V$ containing the $\alpha_i$-partial dominating set (of size $\gamma_{\alpha_i}$) is a max dominating $k$-set.
Hence, if we know an efficient algorithm for the partial dominating set problem, then we can have an efficient algorithm for the max dominating $k$-set problem.

Conversely, if an algorithm for the maximum dominating $k$-set problem is known, we can use it to solve the $\alpha$-partial dominating set problem for a given $\alpha \in (0,1]$ as follows. We solve the maximum dominating $k$-set problem for $k=1,2, \ldots, n$ in order, and stop as soon as the maximum dominated neighborhood size of the obtained $k$-set at some step is greater than or equal to $\lceil{\alpha\cdot n}\rceil$. Report the corresponding max dominating $k$-set as the $\alpha$-partial dominating set for the given $\alpha$. 
\end{proof}

The concept of domination defect of graphs have been studied in \cite{jn1, jn6}, where the $r$-domination defect of a graph $G=(V,E)$ is the minimum number of vertices which are left undominated by a vertex subset of size $\gamma-r$ of $G$ (here, $\gamma$ is the size of minimum dominating set of $G$). The domination defect of some parametrized families of graphs, the composition of graphs, the edge-corona of graphs, and the join and corona of graphs have been studied \cite{jn4,jn2,jn3,jn5}. 

Note that $r$-domination defect of a graph $G$ is nothing but ($|V|$) $-$ (max dominated neighborhood size of a $(\gamma-r)$-set). Hence, by using Theorem \ref{t2}, we also have the following result. 

\begin{corollary}
    Partial domination and domination defect are polynomially equivalent regarding computational complexity.
\end{corollary}

\subsection{Existence of graphs where domination and partial domination are in different complexity classes}
We have already seen that due to Theorem \ref{t1}, the partial domination problem is NP-complete in those graph classes for which the domination problem is so. Now, a natural question arises whether the converse is also true; in other words, whether partial domination is polynomial time solvable whenever domination is so. Our following construction answers this.

Let $\mathcal{C}=c_1\land c_2\land \dots \land c_m$ be a $2$-CNF with $m$ clauses, where there are $n$ variables, namely $x_1,x_2,\dots ,x_n$, used in the clauses.  We construct the graph $G_{\mathcal{C}}=(V_{\mathcal{C}},E_{\mathcal{C}})$ corresponding to this $\mathcal{C}$ as follows. 

The vertex set is $V_{\mathcal{C}}=\{x_1,x_2,\dots,x_n,\bar{x}_1,\bar{x}_2,\dots,\bar{x}_n,$$c_1,c_2,\dots,c_m,d_{1_1},d_{1_2},\dots,$\\$d_{1_{2m}},$ $d_{2_1},d_{2_2},\dots,d_{2_{2m}},\dots,d_{n_1},d_{n_2},\dots,d_{n_{2m}},d\}$. 

The edge set is $E_{\mathcal{C}}=\{(x_i,\bar{x}_i),(x_i,d_{i_j}),(\bar{x}_i,d_{i_j}) : 1\leq i\leq n, 1\leq j\leq 2m\}\cup\{(x_1,d),(\bar{x}_1,d)\}\cup$ \\$\{(c_j,\ell_{j_1}),(c_j,\ell_{j_2}),(c_j,d) : c_j=\ell_{j_1}\vee\ell_{j_2},$ where $\ell_{j_1},\ell_{j_2}\in\{x_1,\dots,x_n,\bar{x}_1,\dots,\bar{x}_n\}; 1\leq j\leq m\}$.

An example has been illustrated in Figure \ref{fig:pic1}, where in a given $2$-CNF expression $\cal C$, the number of clauses ($m$) = 3, and the number of variables ($n$) = 4.\\

\begin{figure}[ht]
    \centering
    \includegraphics[width=1\linewidth]{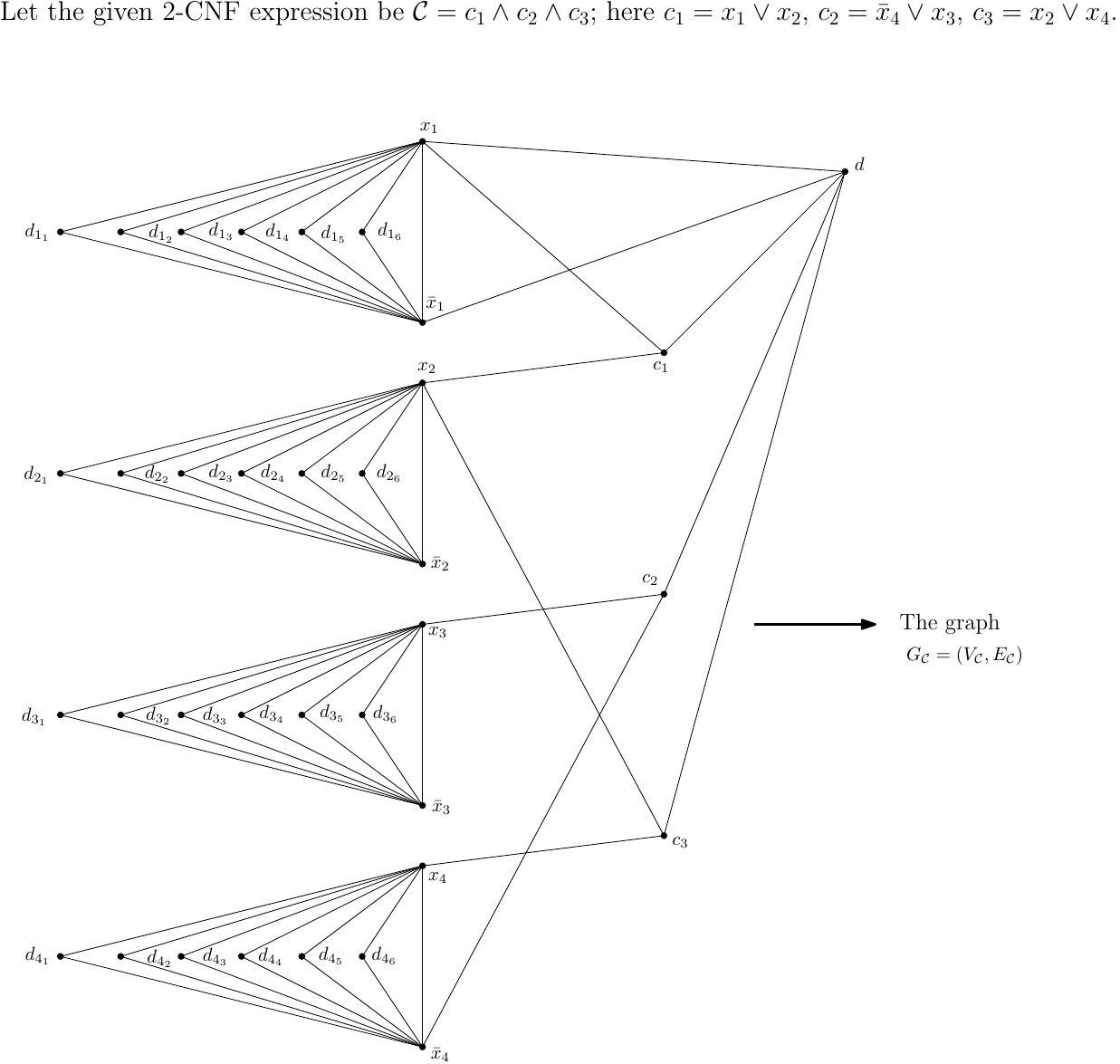}
    \caption{A given $2$-CNF expression $\cal C$ and its corresponding graph $G_{\cal C}=(V_{\cal C}, E_{\cal C})$}
    \label{fig:pic1}
\end{figure}

\begin{lemma}\label{lh1}
    The minimum dominating set problem is polynomial-time solvable for the graph $G_{\mathcal{C}}$.
\end{lemma}

\begin{proof}
    First, observe that any minimum dominating set of the graph $G_\mathcal{C}$ must contain at least one of the vertices between $x_i$ and $\bar{x}_i$ for all $1\leq i\leq n$. 
    
    Because if not, then there exists a minimum dominating set of $G_\mathcal{C}$ for which there is some $i$, such that neither $x_i$ nor $\bar{x}_i$ is in that minimum dominating set. But then to dominate $d_{i_{1}},\dots,d_{i_{2m}}$, it must hold that the said minimum dominating set contains all the vertices $d_{i_{1}},\dots,d_{i_{2m}}$. Instead, if we take $x_i$ and do not take these $2m$ vertices, we still get a dominating set of the graph $G_{\cal C}$. Since $2m > 1$, our claim follows. Thus, the size of the minimum dominating set is greater than or equal to $n$. Also, observe that $\{x_1,x_2,\dots,x_n,d\}$ is a dominating set of the graph $G_\mathcal{C}$. Thus, the size of the minimum dominating set is less than or equal to $(n+1)$. Implying, 
    the size of minimum dominating set of $G_{\mathcal{C}}$ is either $n$ or $(n+1)$.

    Now, observe that if the $2$-CNF $\mathcal{C}$ is satisfiable, then there exists a truth assignment of the $n$ variables for which all the clauses are true. According to this truth assignment, for each $i$, either $x_i$ or $\bar{x}_i$ will be true, not both. We take the vertices corresponding to these true values for all $i \in \{1,\ldots, n\}$, and get a vertex subset of size $n$. Since all the clauses are true for this truth assignment, at least one literal of each clause must be true. Hence, this $ n$-sized subset of vertices must be a dominating set of $G_\mathcal{C}$ due to the construction. So, in this case, the size of the minimum dominating set of $G_\mathcal{C}$ must be $n$.

    Conversely, suppose that the minimum dominating set of $G_\mathcal{C}$ has size $n$. By our first observation, the minimum dominating set contains exactly one of $x_i$ and $\bar{x}_i$ for all $i\in \{1, \ldots, n\}$. We take these $n$ variables (or, their negations) to be true and the rest to be false. This will give a satisfiable assignment of the $2$-CNF $\mathcal{C}$ by our construction.

    Therefore, the $2$-CNF $\mathcal{C}$ is satisfiable if and only if its corresponding graph $G_{\mathcal{C}}$ has domination number $n$.

    Since the $2$-SAT problem can be solved in polynomial time, the result follows.
\end{proof}

\begin{lemma}\label{lh2}
    The partial domination problem is NP-hard in the graph $G_{\mathcal{C}}$.
\end{lemma}

\begin{proof}
    Assume that the partial domination problem is polynomial-time solvable in the graph $G_{\mathcal{C}}$.
    
    Consider the maximum dominating $k$-set problem in $G_{\mathcal{C}}$ for a given $k$ with $1\leq k\leq |V_{\mathcal{C}}|$. By Theorem \ref{t2}, this is solvable in polynomial time.

    In particular, for $k=n$, we can find the maximum dominating $n$-set in polynomial time.

    Observe that, for any $k\in \{1,2, \ldots, n\}$, the maximum dominating $k$-set of $G_{\mathcal{C}}$ must be a subset of $\{x_1,x_2,\dots,x_n, \bar{x}_1, \bar{x}_2, \dots, \bar{x}_n\}$. Again, the construction of $G_{\mathcal{C}}$ suggests that the maximum dominating $k$-set will contain either $x_i$ or $\bar{x}_i$, and not both, for some values of $i$. These, in particular, hold for $k=n$. So, the maximum dominating $n$-set induces a truth assignment to the $n$ variables. By the definition of maximum dominating $n$-set of a graph and due to the construction of $G_{\mathcal{C}}$, this induced truth assignment maximizes the number of true clauses in the $2$-CNF $\mathcal{C}$. 
    
    Note that $\{($the neighborhood size of a maximum dominating $n$-set$) - (2n+2mn+1)\}$ denotes the maximum number of satisfying clauses of $\mathcal{C}$. This implies that there exists a polynomial-time algorithm for solving the Max-$2$-SAT problem. However, Max-$2$-SAT is NP-hard (see \cite{papa}). Thus, the result follows.
\end{proof}

The above two lemmas lead to the following result.

\begin{theorem}\label{t3}
    There exist graphs for which the domination problem is polynomial-time solvable, but partial domination is NP-hard.
\end{theorem}

\section{Max dominating $k$-set problem for unit interval graphs} \label{interval}

A graph $G=(V,E)$ is said to be an interval graph if there exists a layout of a set $\cal I$ of $|V|=n$ intervals such that each node $v_i \in V$ corresponds to an interval in $\cal I$; each interval $v_i\in \cal I$ is specified by a pair of points $(a_i,b_i)$ on a real line, termed as its start- and end-point. Each edge $e=(v_i,v_j) \in E$
implies that the intervals corresponding to $v_i$ and $v_j$ overlap in that layout. An example of an interval graph is given in Figure \ref{fig:pic2}, where we have shown the intervals to be non-overlapping in the plane for clarity, though all of them lie on the horizontal (dashed) line. Testing whether a given graph $G=(V,E)$ is an interval graph needs $O(|V|+|E|)$ time \cite{golumbic}; in the same time, the corresponding interval layout can also be obtained. We can assume that the start- and end-points of the members in $\cal I$ are all distinct. 

\begin{figure}[ht]
    \centering
    \includegraphics[width=0.9\linewidth]{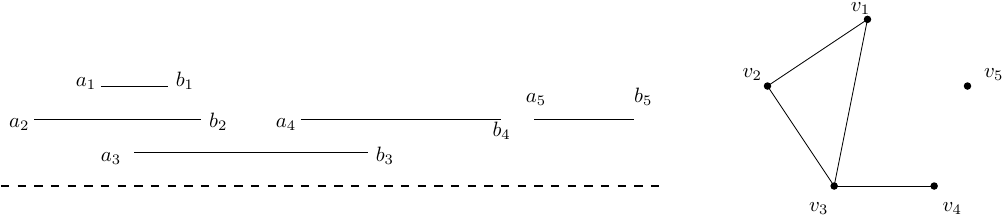}
    \caption{An interval layout and its corresponding interval graph}
    \label{fig:pic2}\vspace{-0.1in}
\end{figure}

In this section,  we provide a dynamic programming algorithm for computing a maximum dominating $k$-set in unit interval graphs, where each interval in the set $\cal I$ has the same length. Here, for each pair of intervals $v_i,v_j\in V$, their corresponding intervals in $\cal I$ are either disjoint or are properly overlapping (none of them is properly contained in the other interval). See Figure \ref{fig:pic3} for an example of a unit interval graph.

\begin{figure}[ht]
    \centering
    \includegraphics[width=0.9\linewidth]{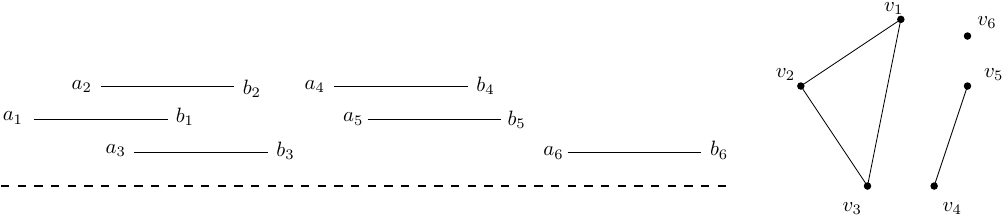}
    \caption{A unit interval graph consisting of $6$ vertices}\vspace{-0.1in}
    \label{fig:pic3}
\end{figure}

\subsection{Dynamic programming algorithm}
{\bf Preprocessing:} We sort the end-points of the intervals in the set $\cal I$, and name the intervals as $v_i=[a_i,b_i]$, $i=1,\dots,n$, according to the increasing order of their right end-points. For each interval $v_i=[a_i,b_i]$, we store the following information: 
\begin{enumerate}
    \item $x(i)=$ Number of $b_j$'s ($j\neq i$) between $a_i$ and $b_i$,
    \item $y(i)=$ Number of $a_j$'s ($j\neq i$) between $a_i$ and $b_i$, 
    \item $z(i)=$ Number of $a_j$'s ($j\neq i$) between $b_{i-1}$ and $b_i$,
    \item $r_a(i)= max\{a_j| a_i\leq a_j<b_i\}$, and
    \item $r_b(i)= b_j$, where $j$ is such that $r_a(i)=a_j$.
\end{enumerate}

{\bf Algorithm:} For a fixed $b_i$, $i\in\{1,\dots,n\}$, and a fixed integer $\ell$ ($\ell\leq k)$, let us define three functions $f$, $g$ and $h$ as follows,
\begin{itemize}
    \item[$\bullet$] $f(b_i,\ell) = \max \{|N[S_\ell]|: S_\ell\subseteq\{v_1,\dots, v_i\}, |S_\ell|=\ell, v_i\in S_\ell$\} (* max dominated {\em nbd} size of an $\ell$-set $S_\ell$, where $S_\ell\subseteq\{v_1,\dots,v_i\}$ and $v_i\in S_\ell$ *), 
    \item[$\bullet$] $g(b_i,\ell)= \max \{ |N[S_\ell]|: S_\ell\subseteq\{v_1,\dots, v_i\}, |S_\ell|=\ell, v_i\notin S_\ell, v_i\in N[S_\ell ]$\} (* max dominated {\em nbd} size of an $\ell$-set $S_\ell\subseteq\{v_1,\dots,v_i\}$, when $v_i\not\in S_\ell$, $v_i\in N[S_\ell]$ *), 
    \item[$\bullet$] $h(b_i,\ell)= \max \{ |N[S_\ell]|: S_\ell\subseteq\{v_1,\dots, v_i\}, |S_\ell|=\ell, v_i\notin S_\ell, v_i\notin N[S_\ell]$\}  (* max dominated {\em nbd} size of an $\ell$-set $S_\ell\subseteq\{v_1,\dots,v_i\}$, when $v_i\not\in S_\ell$, $v_i\not\in N[S_\ell]$ *).
\end{itemize}

The maximum dominated neighborhood size ({\em nbd size}) of a $k$-set is \\
$\max\{f(b_n,k),g(b_n,k),h(b_n,k)\}$.

We find the values of the functions $f$, $g$, and $h$ recursively using dynamic programming. We initialize these functions for $b_1$ as 

\begin{itemize}
\item[$\bullet$] $f(b_1,1)=1+y(1)$ and $f(b_1,j)=-\infty$ for $j=0,2,3, \ldots, k$. 
\item[$\bullet$] $g(b_1,j)=-\infty$ for all $j=0,1,2,3, \ldots, k$.
\item[$\bullet$]  $h(b_1,0)=0$ and $h(b_1,j)=-\infty$ for $j=1,2,3, \ldots, k$.
\end{itemize}
We process $b_i$, $i=1,2,\ldots, n$ in this order. For each $b_i$, we compute $f(b_i,l)$, $g(b_i,l)$ and $h(b_i,l)$, for $l=0,1,2,\ldots, k$, and proceed to compute the $f,g,h$ functions for $b_{i+1}$. For each $b_i$, $f(b_i,0)=g(b_i,0)=-\infty$ (since these are undefined), and $h(b_i,0)=0$. Also, for each $b_i$, we set $f(b_i,1)=1+y(i)+x(i)$. The recursive formulae for finding the values of the functions $f$, $g$, and $h$ are stated below.    
\begin{itemize}
\item[$\bullet$] $f(b_i,l)= \max\{ T_1,T_2,T_3\}$, where
    $$T_1=1+\max_{\forall b_j<a_i, r_b(j)<a_i}\{(f(b_j,l-1)+y(i)+x(i))\}$$
    $$T_2=1+\max_{\forall b_j<a_i, r_b(j)=b_\lambda>a_i}\{f(b_j, l-1)+y(i)+ i-1-\lambda \}$$
    $$T_3=\max_{\forall b_j \in (a_i,b_i)} \{f(b_j,l-1)+z(j+1)+z(j+2)+\ldots +z(i-1)+z(i)\}$$ 
\item[$\bullet$] $g(b_i,l)= \max_{\forall b_j \in (a_i,b_i)} \{f(b_j,l)\}$  
\item[$\bullet$] $h(b_i,l)= \max_{\forall b_j<a_i} \{f(b_j,l)\}$ 
\end{itemize}

Note that, if there does not exist any $b_j\in(a_i,b_i)$ in the recursion formula of $g(b_i,l)$, then the value of $g(b_i,l)=-\infty$ is set, where $l\in\{1,2,\dots ,k\}$. Similarly, when there does not exist any $b_j$ satisfying $b_j<a_i$ in the recursion formula of $h(b_i,l)$, then $h(b_i,l)=-\infty$ is set.

{\bf Justifications of the formulae:}

$\mathbf{f(b_i,l)}$: For the function $f$, the interval $v_i$ is to be chosen, so we have the following three possibilities,

\begin{itemize}
    \item $\mathbf{T_1}$: In this case, $v_i$ is not already dominated by the last chosen vertex $v_j$, so 1 is added. Here, additionally, because $r(b_j)<a_i$, we have to add $x(i)$ (number of intervals that ends in $[a_i,b_i]$) and $y(i)$ (number of intervals that starts in $[a_i,b_i]$), which are not dominated by any other previously chosen interval. 

    \item $\mathbf{T_2}$: Here also, $v_i$ was not dominated by $v_j$, so 1 is added. Now, since $r(b_j)=b_\lambda \in (a_i,b_i)$, choosing $v_i$ will newly dominate $y(i)$ (number of intervals that starts in $[a_i,b_i]$), and $(i-1-\lambda)$  (the number of intervals that ends in $(b_\lambda,b_i)$) nodes, which are not dominated by the previous chosen interval $v_j=[a_j,b_j]$. 

\item $\mathbf{T_3}$:  Here, the endpoint $b_j$ of the last chosen interval $v_j$, lies in $[a_i,b_i]$; thus, $v_i$ is already dominated by $v_j$. So unlike $T_1$ and $T_2$, $1$ is not added. As $a_j<a_i$, whatever end-points and start-points exist inside $[a_i,b_j]$, their corresponding intervals are already dominated and have been counted in $f(b_j,l-1)$. So, we need to consider the part $[b_j,b_i]$. In between $b_j$ and $b_i$, the only possible endpoints are $b_{j+1},b_{j+2},\dots,b_{i-1},b_{i}$. But their start points are between $a_j$ and $a_i$. So, they are already dominated by the interval $[a_j, b_j]$ and have been counted before. So, we only need to add the number of intervals having their start point between $b_j$ and $b_i$, which are $z(j+1)+z(j+2)+\dots+z(i-1)+z(i)$. This is because they overlapped with $v_i$ and had not been counted yet. 
\end{itemize}

$\mathbf{g(b_i,l)}$: In this case, the interval $v_i$ is not chosen but was already dominated by a previously chosen interval $v_j$. It is indicated by $b_j$'s range in the $\max$ operation.
  
$\mathbf{h(b_i,l)}$: The interval $v_i$ is neither chosen nor was dominated by any previously chosen interval $v_j$, and the range of $b_j$ indicates this in the $\max$ operation of the recursion formula.

\subsection{Processing, correctness and complexity results} 

During the preprocessing, after sorting, we compute $x(i)$, $y(i)$, $z(i)$, $r_a(i)$ and $r_b(i)$ by processing the end-points of the intervals in $\cal I$ in left to right order. We conceptually maintain a queue $Q$ for storing the processed subset of left end-points of the intervals whose right end-points are not yet processed. Practically, we maintain three scalar variables, $Q\_count$ (size of $Q$), $Q\_max$ (maximum element in $Q$), and $Q^r\_max$ (the right-end point of the interval corresponding to $Q\_max$). While processing a left-end point, say $a_j$, we add 1 to $Q\_count$, store $a_j$ in $Q\_max$, and its corresponding $b_j$ in $Q^r\_max$. When a right-end point $b_i$ of an interval $v_i$ is faced, its corresponding left-end point is the first element of $Q$ (as all the intervals are of unit length); it is deleted, or equivalently, $1$ is subtracted from $Q\_count$. Note that, $y(i)$ is the size of $Q$ at this instant and $r_a(i)$ is the maximum element in $Q$; thus we set $y(i)=Q\_count$, $r_a(i)=Q\_max$, and $r_b(i)=Q^r\_max$. If at this instant, $Q$ is empty queue, then set $r_a(i)=a_i$ and $r_b(i)=b_i$. The $z(i)$ values are also computed in the same left-to-right scan by maintaining another scalar variable $z\_count$.  Another right-to-left scan is needed to compute $x(i)$ values for each interval $v_i\in \cal I$. Thus, the preprocessing time complexity is $O(n\log n)$ using $O(1)$ extra space. In dynamic programming, $T_1,T_2$ and $T_3$ are computed as follows: 

\begin{itemize}

\item While computing $f(b_i,l)$, note that $r_b(i) \geq b_i$ for all $i$. Thus, $T_1=1+y(i)+x(i)+\max_{\forall r_b(j)<a_i}f(b_j,l-1) = 1+y(i)+x(i)+T_1'$ (say). We maintain a leaf search height-balanced binary tree $D_1$ on $\{r_b(j), j=1,2,\ldots, n\}$, where each leaf is associated with $f(b_j,l-1)$, and each internal node associated with the maximum of the $f(.,l-1)$ values in the subtree rooted at that node. Thus, we can compute $T_1'=\max_{\forall r_b(j)<a_i}f(b_j,l-1)$ in $O(\log n)$ time using $D_1$.

\item For computing $T_2$, we need to compute $T_2'=\max_{\forall b_j<a_i, a_i<r_b(j)=b_{\lambda}<b_i}(f(b_j, l-1)-\lambda)$, where $b_\lambda=r_b(j)$. Again, maintaining a similar leaf search height-balanced binary tree $D_2$ on $\{r_b(j),j=1,2,\ldots, n\}$ with each leaf node associated with $(f(b_j,l-1)-\lambda)$, and each internal node associated with the maximum of the $(f(.,l-1)-\lambda)$ values in the subtree rooted at that node, one can compute $T_2'$ in $O(\log n)$ time. Now, $T_2=1+y(i)+i-1+T_2'$ is computed using $T_2'$, obtained earlier.

\item For computing $T_3$, we use a $D_3$ data structure, which is a height-balanced leaf-search binary tree on $\{b_j, j=1,2,\ldots, n\}$, whose each leaf $b_\alpha$ contains $val_\alpha=f(b_\alpha,l-1) +n_\alpha$, where $n_\alpha$ = number of start-points in the interval $[b_\alpha,b_i]$. Each internal node whose sub-tree contains leaves $\{b_\alpha, \ldots ,b_\beta\}$ represents a $\max$-field containing the maximum of $\{val_\alpha, \ldots, val_\beta\}$. Now, the $T_3$ value for $b_i$ is obtained by searching the $D_3$ data structure for all $b_j$ satisfying $a_i<b_j<b_i$. This also can be done in $O(\log n)$ time. 

\end{itemize}

Thus, the computation of $f(b_i,l)=\max\{T_1, T_2, T_3\}$ needs $O(\log n)$ time using three data structures, each of size $O(n)$, and can be created in $O(n\log n)$ time.  Similarly, $g(b_i,l)$ and $h(b_i,l)$ can also be computed in $O(\log n)$ time using a $O(n)$ size data structure which can be built in $O(n\log n)$ time. Thus, 
\begin{theorem}
The maximum dominating $k$-set problem algorithm for unit interval graphs runs in $O(nk\log n)$ time using data structures of size $O(n)$. 
\end{theorem}

\section{Algorithm for general interval graph}\label{arbint}

Now, we relax the unit-interval assumption. Here, a pair of overlapping intervals in $\cal I$ may either properly overlap or one interval properly contains the other one. We first execute a preprocessing phase based on the following observation.

\begin{observation}
Let $v_i, v_j\in V$ be two vertices in $G$ such that, in $\cal I$,  the interval corresponding to $v_i$ is completely contained in the interval corresponding to $v_j$. Here, if there exists an optimum solution $S_{opt}$ of the maximum dominating $k$-set problem with $v_i\in S_{opt}$ then there exists another optimum solution $S_{opt}'$ of that problem with $v_j\in S_{opt}'$ and $|S_{opt}|=|S_{opt}'|$.    
\end{observation}

\begin{figure}[ht]
    \centering
    \includegraphics[width=0.9\linewidth]{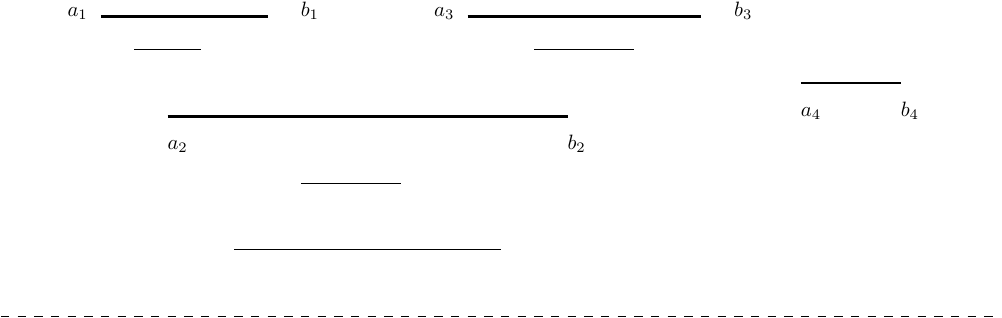}
    \caption{Here, $[a_1,b_1],[a_2,b_2],[a_3,b_3]$ and $[a_4,b_4]$ are non-deleted intervals, and the rest are marked as deleted intervals. All these intervals lie on the real line (dashed line), but they have been drawn in the two-dimensional plane for visual clarity.}
    \label{fig:pic4}
\end{figure}

Thus, we execute a linear pass to identify the intervals of $\cal I$ which are contained in at least one other interval of $\cal I$. These intervals are marked as deleted (see the intervals drawn with thin lines in Figure \ref{fig:pic4}).

Observe that the graph corresponding to the non-deleted intervals can be represented as a unit interval graph. Thus, we interchangeably use the terms unit interval layout and properly overlapping interval layout.
During this processing, for each non-deleted interval $v_i=[a_i,b_i]$, $i\in\{1,2,\dots ,m\}$ $(m\leq |V|=n)$, we also compute the following information: 

$x_d(i)=$ Number of deleted intervals with only right-end point  in $[a_i,b_i]$,\\
$y_d(i)=$ Number of deleted intervals with only left-end point in $[a_i,b_i]$, \\
$w_d(i)=$ Number of deleted intervals with both the end points in $[a_i,b_i]$.

Also, for each pair of intervals $v_i$ and $v_j$, we find and store $c_d(i,j)=$ number of deleted intervals dominated by both $v_i$ and $v_j$.

From now onwards, we use $\cal I'$ to denote the undeleted intervals. The end-points of the members in $\cal I$ are stored in a sorted list $\cal L$, where each end-point of an interval $I\in \cal I$ points to its other end-point in $\cal L$. We execute a linear pass among the members in $\cal L$. During this execution, we maintain a balanced leaf-search binary tree $\cal T$. When the left-end point of a member of $\cal I$ is processed, it is inserted in $\cal T$. When the right-end point $b$ of an interval $I=[a,b]$ is processed, we search $a$ in $\cal T$. If at least one entry to the left of $a$ in $\cal T$ exists, then $I$ is fully contained in that interval; $I$ is marked deleted and $a$ is deleted from $\cal T$. 
Thus, in $O(n\log n)$ time we can compute the set ${\cal I}'$ of properly overlapped intervals in the set $\cal I$. The set ${\cal I}\setminus {\cal I}'$ is the set of deleted intervals, which will not be considered in our dynamic programming algorithm for computing the maximum dominating $k$-set. To compute the values of the parameters defined earlier, we create a 2D range tree $\cal R$, whose each element is a 2D point $(\alpha,\beta)$, which corresponds to the two end-points of a deleted interval $[\alpha,\beta]$ ($\in {\cal I}\setminus {\cal I}'$). We consider each interval $I_i=[a_i,b_i] \in {\cal I}'$. Here, (i) $x_d(i)$ = number of $\beta$-points lying in $I_i$ whose $\alpha$-point does not lie in $I_i$, (ii) $y_d(i)$ =  number of $\alpha$-points lying in $I_i$ whose $\beta$-point does not lie in $I_i$, and (iii) $w_d(i)$ = number of points whose both $\alpha$ and $\beta$ values lie in $I_i$. All these three numbers can be computed in $O(\log^2n)$ time.

To compute $c_d(i,j)$, for $v_i,v_j \in {\cal I}'$, we create a bipartite graph $B({\cal I}\setminus {\cal I}', {\cal I}')$. While computing $x_d(i)$, $y_d(i)$ and $w_d(i)$, we have already identified the deleted intervals in ${\cal I}\setminus {\cal I}'$ that are dominated by each interval $I_i \in {\cal I}'$. Those edges are put in the graph $B$. At the end of this processing, we observe each node $u_\theta\in {\cal I}\setminus {\cal I}'$, and add 1 to each cell $c_d(i,j)$ such that $u_\theta$ is adjacent to $v_i,v_j \in {\cal I}'$ (if any such pair exists). Thus, the worst-case time complexity of computing the graph $B$ is $O(n^2)$, and computing $c_d(i,j), v_i,v_j \in {\cal I}'$ is also $O(n^2)$. 

Note that the $m$ intervals of ${\cal I}'$ are properly overlapping intervals. So, this set of $m$ intervals behaves as unit intervals, and computing the functions $x(i)$, $y(i)$, and $r_b(i)$ for $i=1,2,\dots,m$ is as in the previous section. 
Also, the definitions of the functions $f$, $g$, and $h$ remain the same as in the previous section; their recursive formulae are given below. Finally, the maximum dominated neighborhood size of a $k$-set is computed as $\max\{ f(b_m,k),g(b_m,k),h(b_m,k)\}$.

{\bf{Initialization:}} The values of the three functions for $b_1 \in {\cal I}'$ are as follows,

\begin{itemize}
\item[$\bullet$] $f(b_1,1)=1+y(1)+x_d(1)+y_d(1)+w_d(1)$ and $f(b_1,j)=-\infty$ for $j \neq 1$,
\item[$\bullet$] $g(b_1,j)=-\infty$ for all $j=0,1,\ldots, k$, and
\item[$\bullet$] $h(b_1,0)=0$ and $h(b_1,j)=-\infty$ for $j=1,2,\ldots, k$.
\end{itemize}

Also, we have $f(b_i,0)=-\infty$, $g(b_i,0)=-\infty$ and $h(b_i,0)=0$ for all $i\in \{1,2,\dots,m\}$. For all $b_i$, $f(b_i,1)=1+y(i)+x(i)+x_d(i)+y_d(i)+w_d(i)$ is also set in this initialization step.

{\bf{Recursive formulas:}} 

Recursive formulas for finding the values of the functions $f$, $g$, and $h$ are as follows, which are minor tailoring of the formulae presented for unit intervals, and the justification of the tailoring is also mentioned.

\begin{itemize}
\item[$\bullet$] $f(b_i,l)=\max\{T_1, T_2, T_3\}$, where\\
    $T_1=1+\max_{b_j<a_i, r_b(j)<a_i} \{f(b_j,l-1)+y(i)+x(i)+x_d(i)+y_d(i)+w_d(i)-c_d(i,j)\}$, \\
    $T_2= 1+\max_{b_j<a_i, r_b(j)=b_{\lambda}>a_i}\{f(b_j, l-1)+y(i)+i-1-\lambda+x_d(i)+y_d(i)+w_d(i)-c_d(i,j)\}$, and \\
    $T_3=\max_{a_i<b_j<b_i} \{f(b_j,l-1)+z(j+1)+z(j+2)+\dots +z(i-1)+z(i)+x_d(i)+y_d(i)+w_d(i)-c_d(i,j)\}$,\\
\item[$\bullet$] $g(b_i,l)= \max_{a_i<b_j<b_i}\{f(b_j,l)\}$,\\
\item[$\bullet$] $h(b_i,l)= \max_{b_j<a_i}\{f(b_j,l)\}$.
\end{itemize}

If there does not exist any $b_j\in(a_i,b_i)$ in the recursive formula of $g(b_i,l)$, then $g(b_i,l)=-\infty$ is set, where $l\in\{1,2,\dots ,k\}$. Also, when no $b_j$ with $b_j<a_i$ exists in the recursive formula of $h(b_i,l)$, then $h(b_i,l)=-\infty$, where $l\in\{1,2,\dots ,k\}$.

{\bf{Justification of the recursive formulas:}}

When we consider the $m$ non-deleted intervals, which are dominated by the chosen $l$-set, the recursive formulae will be the same as those of the unit interval graph on those $m$ intervals. Below, we show that the count of the dominated deleted intervals by a chosen $l$-set is correct.

As the interval $v_i\in {\cal I}'$ is included in the chosen $l$-set, the number of deleted intervals dominated by  $v_i$ is $x_d(i)+y_d(i)+w_d(i)$; among them $c_d(i,j)$ = number of deleted intervals which are already dominated by the last included interval $v_j=[a_j,b_j]\in {\cal I}'$, was previously counted in $f(b_j,l-1)$. Hence, we added $x_d(i)+y_d(i)+w_d(i)-c_d(i,j)$ in the recursion formula of $f(b_i,l)$, which is precisely the number of newly dominated deleted intervals by $v_i$.

For the functions $g$ and $h$, the $i$-th interval is not included in the $l$-set while processing. Hence, no new deleted interval is dominated. 
Thus, the $O(n^2)$ time preprocessing leads to the following result.

\begin{theorem}
The time complexity for computing the maximum dominating $k$-set problem for the intersection graph of a set of arbitrary intervals is $O(n^2k)$. 
\end{theorem}

\section{Unit square intersection graph where a straight line intersects each square} \label{square}
The partial domination problem is NP-hard for axis-parallel unit square intersection graphs, since the minimum dominating set problem is NP-hard in this class. So, we will explore a special subclass of this class.

Consider the intersection graph of a given set $S=\{s_1,s_2,\dots,s_n \}$ of axis parallel unit squares in $I\!\!R^2$, where all the members of $S$ are intersected by the straight line $L: y=mx+c$, that makes an angle $\theta$ with the positive $x$-axis (i.e., $m=\tan\theta$), where $0\degree \leq \theta\leq 45\degree$. The objective is to compute the maximum dominating $k$-set for the intersection graph of members of $S$. 

\begin{figure}[ht]
    \centering
    \includegraphics[width=0.9\linewidth]{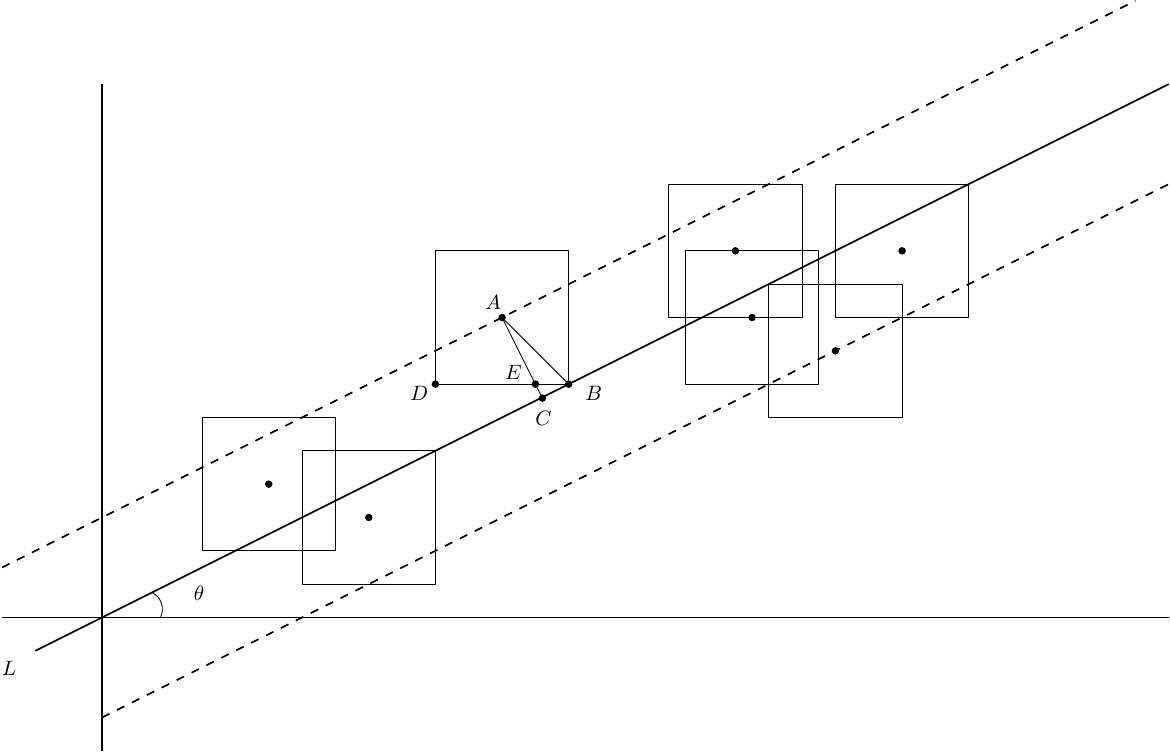}
    \caption{Here, the width of the required strip is $2AC$.}
    \label{fig:pic5}
\end{figure}

\begin{observation}\label{obx1}
    If $T$ is a strip (region bounded by two parallel lines) which is divided into two equal-width strips by the line $L$, and centers of all the squares in $S$ lie inside $T$, then the minimum width of $T$ is $\sqrt{2}\sin(45\degree+\theta)$,
\end{observation}

\begin{proof}
    Consider Figure \ref{fig:pic5}. Since the input squares are unit squares, we have $AB=\frac{1}{\sqrt{2}}$.

    Now, $\angle CBD=\theta$ and $\angle ABE=45\degree$; implying 
    $\angle ABC=45\degree+\theta$.
    
     As $\sin (45\degree+\theta)=\sin\angle ABC=\frac{AC}{AB}=\sqrt{2}AC$, we have $AC=\frac{\sin(45\degree+\theta)}{\sqrt{2}}$.
    
    $\therefore$ Width of the strip  $=2\cdot\frac{\sin(45\degree+\theta)}{\sqrt{2}}=\sqrt{2}\sin(45\degree+\theta)$.
\end{proof}
    
Since the members in $S$ are axis-parallel and intersected by $L$, here a pair of unit squares in $S$ intersect if and only if they have at least one intersection point inside the strip $T$ of width $\sqrt{2}\sin(45\degree+\theta)$ (see Observation \ref{obx1}).

We now divide the strip $T$ into square-sized boxes of side-length $\sqrt{2}\sin(45\degree+\theta)$ (see Figure \ref{fig:pic6}). As the members in $S$ are unit squares, each $s_i\in S$ may have a non-empty intersection with at most two consecutive boxes in $T$. Thus, at most $2n$ boxes may contain some portion of the squares in $S$. Let us denote these boxes as $T_1, T_2, \dots, T_m$ ($m\leq 2n$). We use $S_i(\subseteq S)$ to denote the squares whose centers lie in $T_i$.  

\begin{figure}[ht]
    \centering
    \includegraphics[width=0.9\linewidth]{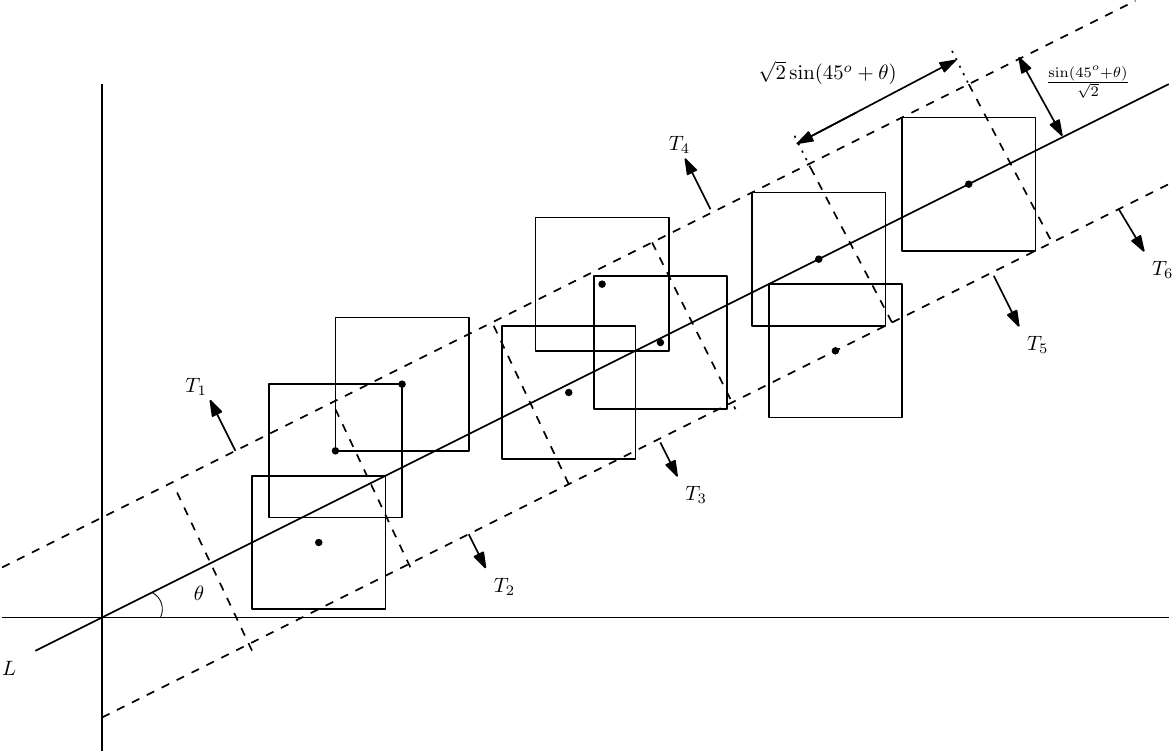}
    \caption{Demonstration of the problem}
    \label{fig:pic6}
\end{figure}

\begin{result} \label{rr} \cite{new1} 
For each $i \in \{1,2,\dots,m\}$, (a) there exists at most $4$ squares of $S$ which can dominate all the squares whose centers lie in $T_i$, and (b) the minimum number of members in $S$ that are needed to dominate all the squares intersecting the box $T_i$ is at most $12$.
\end{result}

Result \ref{rr}(a) leads to the following lemma, which determines the time complexity of our algorithm for the maximum dominating $k$-set problem for $S$.  

\begin{lemma}\label{LLL}
There exists a maximum dominating $k$-set, say $OPT$, where each box $T_i$ contains the centers of at most $11$ members of $OPT$. 
\end{lemma}

\begin{proof}[By contradiction]
Let $OPT$ be an optimum solution to the said problem with at least one box in $T$ having the centers of more than $11$ squares of $OPT$. Let $T_j$ ($j\geq 1$) be the leftmost among such boxes containing the centers of the squares $O_{j} \subseteq OPT$, $|O_j|\geq 12$ (here, $O_j=S_j\cap OPT$). By Result \ref{rr}(a), we can replace/choose at most $4$ squares which can dominate all the squares with center inside $T_j$. Thus, the remaining at least $8$ squares in $O_j$ are used to dominate squares centered inside $T_{j-1} \cup T_{j+1}$.
By Result \ref{rr}(a), instead of these remaining squares of $O_j$, we can take $4$ squares from $S_{j-1}$ and $4$ from $S_{j+1}$ to dominate all the squares centered in $T_{j-1} \cup T_{j+1}$. Also, by pigeonhole principle, among the remaining at least $8$ squares in $O_j$, at least $4$ will contribute in dominating elements of $S_{j-1}$ (or $S_{j+1}$).

Without loss of generality, assume that $O_j'\subseteq O_j$ ($|O_j'|\geq 4$) is used to dominate the squares in $S_{j-1}$. We delete $O_j'$ from $OPT$, and add at most $4$ squares, say $O_j''$ from $S_{j-1}$, centered at four sub-boxes of $T_{j-1}$, in $OPT$ to dominate all the squares in $S_{j-1}$. These, in addition, may dominate some squares in $S_{j-2}$. Thus, we now have $\tilde{O}_{j-1}=O_{j-1}\cup O_j''$, and the total size of $OPT$ is not increased. Also, $|O_j|$ became $<12$. If $|\tilde{O}_{j-1}|\geq 12$, then again, we apply the same technique to delete some squares from $\tilde{O}_{j-1}$ and add the required number ($\leq 4$) squares from $S_{j-2}$ without increasing the initial size of $OPT$. The process continues up to the box $T_1$. If the size of $\tilde{O}_1$ (after adding some new squares) is $\geq 11$, we can replace it by at most $4$ squares. 
Thus, the size of $OPT$ is reduced, maintaining the size of the dominated squares in $S_1\cup \ldots, \cup S_j$ unchanged or increased, while the number of squares in $O_{j+1} \cup \ldots\cup O_m$ remains unchanged. Thus, we have a contradiction that the chosen solution $OPT$ is not optimal.
\end{proof}

Using Lemma \ref{LLL}, we construct a dynamic programming-based algorithm for the maximum dominating $k$-set problem.
\begin{itemize}
    \item $S_i^l$: a subset of $S_i$ of size $l$.
    \item $d(S_i^l)$: set of all squares dominated by $S_i^l$.
    \item $d_c(S_i^{l_1}\cup S_{i+1}^{l_2},S_{i+2}^{l_3})$: set of all squares dominated by both $S_i^{l_1}\cup S_{i+1}^{l_2}$ and $S_{i+2}^{l_3}$.
\end{itemize}
For a specific choice of the tuple $(l_1,l_2,S_{i-1}^{l_3},S_i^{l_4})$, where $l_1$, $l_2$, $l_3$ and $l_4$ are non-negative integers, $S_{i-1}^{l_3}$ is a subset of size $l_3$ of $S_{i-1}$ and $S_i^{l_4}$ is a subset of size $l_4$ of $S_i$, we denote by $N(l_1,l_2,S_{i-1}^{l_3},S_i^{l_4})$ the maximum dominated neighborhood size of a $(l_1+l_2+l_3+l_4)$-set, where $S_{i-1}^{l_3}$ is a subset of $l_3$ elements chosen from $S_{i-1}$, $S_i^{l_4}$ is a subset of $l_4$ elements chosen from $S_i$, a subset of $l_2$ elements is chosen from $S_{i-2}$ and a  subset of $l_1$ elements is chosen from $\cup_{j=1}^{i-3} S_j$ in the $(l_1+l_2+l_3+l_4)$-set.
With these notations, the recursion formula for $N(l_1,l_2,S_{i-1}^{l_3},S_i^{l_4})$ is

$N(l_1,l_2,S_{i-1}^{l_3},S_i^{l_4}) = 
\max \{ N(l_1-a,a,S_{i-2}^{l_2},S_{i-1}^{l_3}) + d(S_i^{l_4}) - d_c(S_{i-2}^{l_2}\cup S_{i-1}^{l_3},S_i^{l_4}): S_{i-2}^{l_2}\subseteq S_{i-2}$ and $|S_{i-2}^{l_2}|=l_2$, and $0\leq a \leq \min\{l_1,11\}$ is an integer$\}$.

Lemma \ref{LLL} says that $l_2,l_3,l_4\leq 11$ for any choice of the tuple $(l_1,l_2,S_{i-1}^{l_3},S_i^{l_4})$. The choice for $l_1$ can be any integer in $\{0,1,\dots,k\}$.

In the preprocessing step of the dynamic programming algorithm, the values of the functions $d$ and $d_c$ are calculated for all the required subsets of $S$. Next, each box $T_i$ is processed for $i=1,2,\dots,m$, in this order. We calculate and store all possible values of $N(l_1,l_2,S_{i-1}^{l_3},S_i^{l_4})$ using the recursion formula. Finally, the optimum maximum dominated neighborhood size $OPT$ of a $k$-set is obtained as 

$OPT = \max \{ N(l_1,l_2,S_{m-1}^{l_3},S_{m}^{l_4}): S_{m-1}^{l_3}\subseteq S_{m-1}$ with $|S_{m-1}^{l_3}|=l_3$, $S_m^{l_4}\subseteq S_m$ with  $|S_m^{l_4}|=l_4$, and $l_1+l_2+l_3+l_4=k, 0\leq l_2, l_3, l_4\leq 11 \}$.

Note that for the initialization of the dynamic programming, we add four dummy boxes $T_{-3},T_{-2},$ $T_{-1}$ and $T_0$ respectively just before the box $T_1$ and initialize $N(0,0,\emptyset,\emptyset)=0$ and all other entries which are undefined are initialized to $-\infty$.

\textbf{Justification of the recursion formula:} 

While processing $T_i$, when we choose $l_4$ elements of the set $S_{i}^{l_4}$ from $i$-th box $T_i$, it dominates $d(S_{i}^{l_4})$ elements. Now, there are only three possibilities,

\begin{enumerate}
    \item None of these $d(S_i^{l_4})$ elements were previously dominated. In this case, the last chosen squares can have centers in at most $T_{i-3}$; not in any later box.
    
    \item Some of these $d(S_i^{l_4})$ elements were already dominated, but none of the $l_4$ elements of $S_i^{l_4}$ were dominated. In this case, the last chosen squares have centers in the box $T_{i-2}$.
    
    \item Some of the $l_4$ elements of $S_i^{l_4}$ were already dominated. Here, the last chosen squares have centers in the box $T_{i-1}$.
\end{enumerate}

Correctness follows from the fact that the above three cases have been taken care of in the 
recursion formula.

\textbf{Time complexity:} Calculating $N(l_1,l_2,S_{i-1}^{l_3},S_i^{l_4})$ for a given $(l_1,l_2,S_{i-1}^{l_3},S_i^{l_4})$ requires $O(n^{11})$ time. For each $i$, number of tuples of the form $(l_1,l_2,S_{i-1}^{l_3},S_i^{l_4})$ is $O(kn^{22})$ and there are at most $2n$ possible values of $i$. Hence, the time needed to calculate the function $N$ for all possible arguments is $O(kn^{34})$. Finally, calculating $OPT$ requires $O(n^{11}.n^{11})=O(n^{22})$ time. Hence, the time complexity of the algorithm is $O(kn^{34})+O(n^{22}) = O(kn^{34})$, a polynomial in $n$ (since $k\leq n$).

\begin{figure}[ht]
    \centering
    \includegraphics[width=0.66\linewidth]{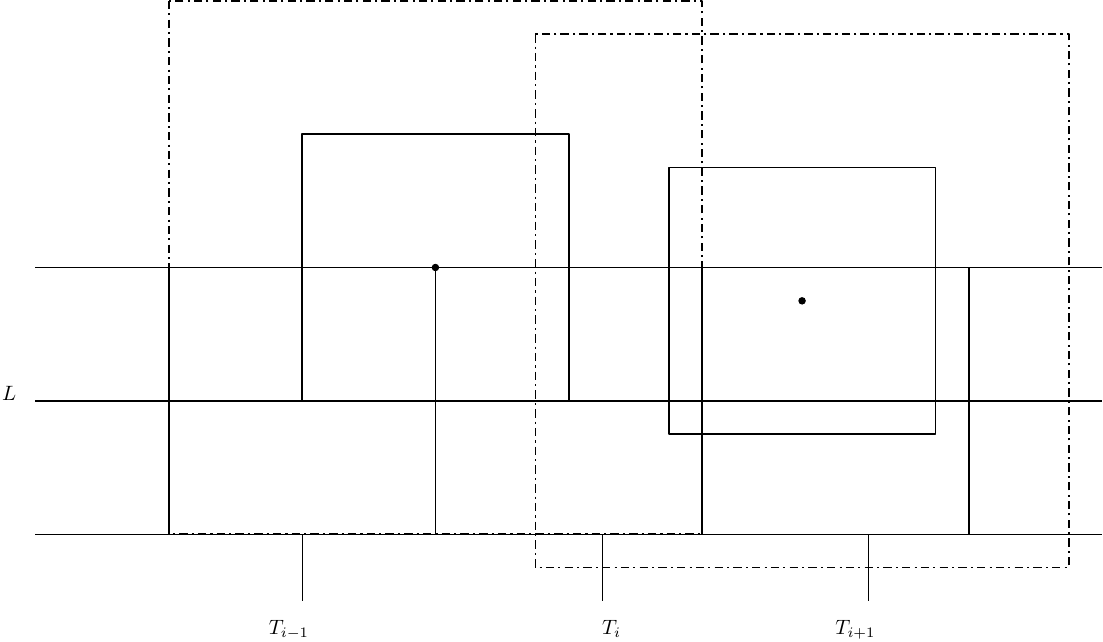}
    \caption{Forbidden regions of unit squares when $\theta=0\degree$}
    \label{fig:pic7}
\end{figure}

\begin{remark}
    Observe that when $\theta=0\degree$, then each box $T_i$ cannot contain centers of two or more non-intersecting squares (See Figure \ref{fig:pic7}). Because for any one square with center in $T_i$, the forbidden region (i.e., the region where the center of any square which does not intersect the given square cannot lie) of that square fully covers the box $T_i$. Thus, in this case, the time complexity of the dynamic programming algorithm reduces to $O(kn^7)$. 
    
    Further, for this case (i.e., when $\theta=0\degree$), the intersection graph of the input unit squares is nothing but the unit interval graph where the unit intervals are the projections of the input squares on the line $L$. So, the maximum dominating $k$ set can be found in $O(nk\log n)$ time in this case. Hence, the time complexity further reduces.
\end{remark}

\begin{figure}[ht]
    \centering
    \includegraphics[width=0.9\linewidth]{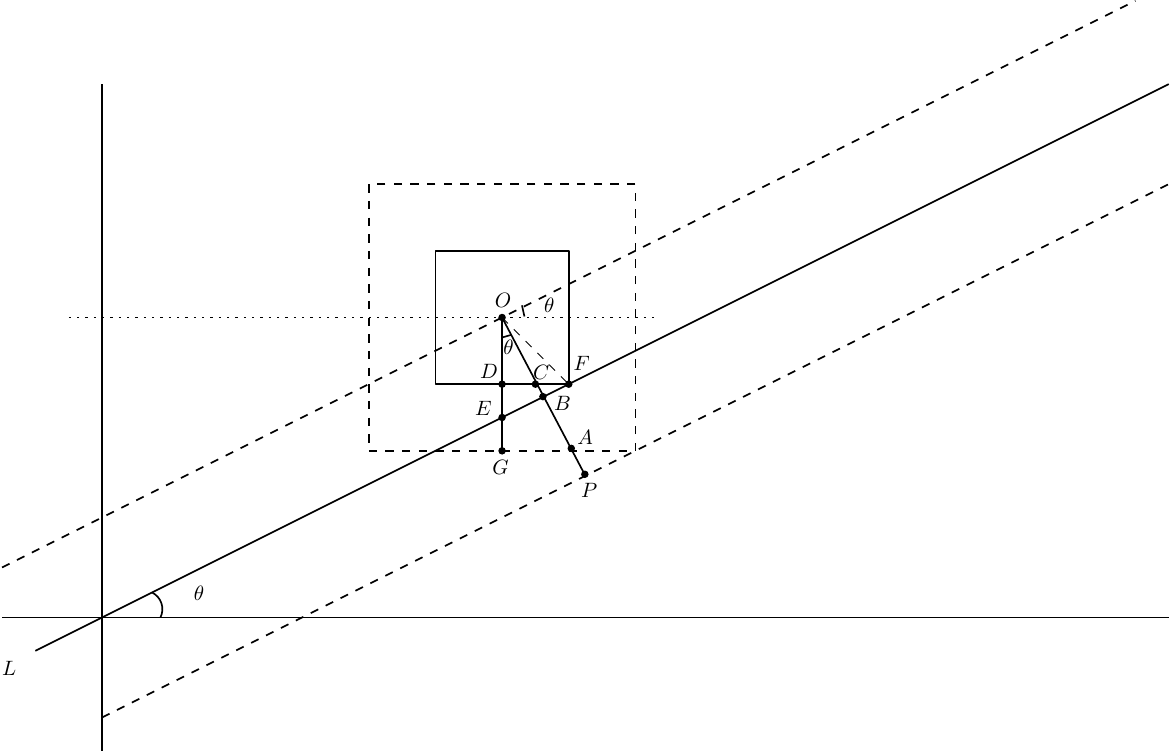}
    \caption{When $0\degree<\theta<45\degree$, $OA<OP$ always holds.}
    \label{fig:pic8}
\end{figure}

\begin{remark}
    When $0\degree<\theta<45\degree$, a box $T_i$ can contain centers of $4$ non-intersecting squares.
\end{remark}

\begin{proof}
    See Figure \ref{fig:pic8}. 
    Here, $h=$ side length of $T_i$ $=OP=2\cdot OB=\sqrt{2}\cdot\sin(45\degree+\theta)$.

    The length of the forbidden region along the line $OP$ is $OA=l$ (say), where we assume $OP$ is one side of the box $T_i$.
    
    Now, $\cos\angle AOG=\frac{OG}{OA}=\frac{1}{OA}$
     $\implies$ $\cos\theta=\frac{1}{OA}$.
     Therefore, we have the following two equations: (i) $l=\frac{1}{\cos\theta}$, and (ii) $h=\sqrt{2}\cdot\sin(45\degree+\theta)$,
     where $\theta\in [0\degree,45\degree]$.

     But, $\forall\theta$ with $0\degree<\theta<45\degree$, we have $l<h$. 
     Hence, a box $T_i$ can contain centers of $4$ non-intersecting squares (See Figure \ref{fig:pic9}).
\end{proof}

\begin{figure}[ht]
    \centering
    \includegraphics[width=0.6\linewidth]{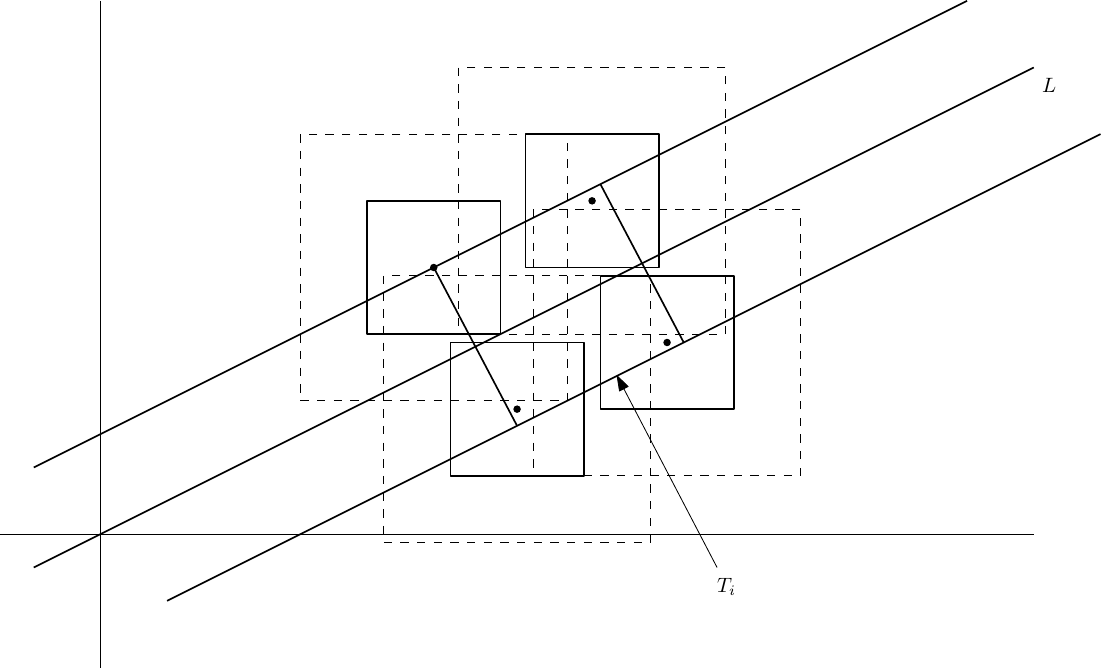}
    \caption{A box $T_i$ can contain centers of $4$ non-intersecting squares.}
    \label{fig:pic9}\vspace{-0.1in}
\end{figure}

\begin{remark}
    When the angle $\theta$ (the angle between the straight line $L$ and positive $x$-axis) does not lie in the range $0\degree\leq\theta\leq 45\degree$, the problem can be dealt similarly, because all other cases can be reduced to this case (sometimes considering angle with $y$-axis instead of $x$-axis).
\end{remark}

\section{Axis parallel unit height rectangles intersected by a straight line}
\label{rectangle}

We now consider the intersection graph of unit-height (assuming width at least 1 unit) axis-parallel rectangles intersected by a straight line $L$. We consider the two cases: (i) $L: y=x$, and (ii) $L: y=mx$, where $m=\tan\theta$, $ 0\degree<\theta<45\degree$ separately.

\subsection{Case (i) $L:y=x$}
Let $R=\{\rho_1,\rho_2,\dots,\rho_n\}$ be the set of input rectangles and the intersecting line is $L:y=x$. 
Consider a strip $T$ of width $\sqrt{2}$ which is divided into two equal-width strips by the straight line $L$. Now, by Lemma 3 of \cite{new1}, two rectangles $\rho_i$ and $\rho_j$ intersect if and only if there exists an intersection point of them inside the strip. Also, the rectangles can have only six types of portions inside the strip (see \cite{new1}). 

Observe that each rectangle contains an unit square whose center lies inside the strip and the unit square corresponding to each rectangle intersects the line $L$. The center of such a unit square is considered to be the center of the corresponding rectangle as well as the center of the corresponding portion. Thus, with each rectangle $\rho_i\in R$, a center is being associated (See Figure \ref{fig:pic10}). 

\begin{figure}[ht]
    \centering
    \includegraphics[width=0.75\linewidth]{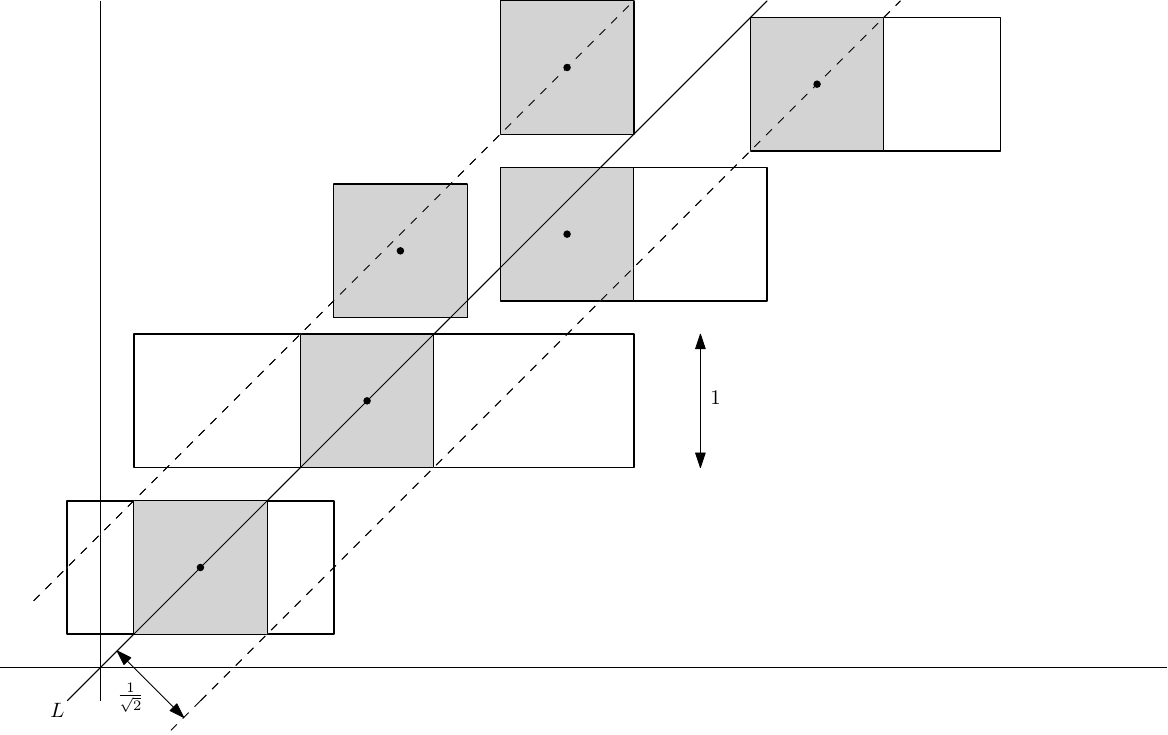}
    \caption{The rectangles, their corresponding unit squares, and their centers}
    \label{fig:pic10}\vspace{-0.1in}
\end{figure}

With this setup, we divide $T$ into rectangular boxes where the width of each rectangular box is $2\sqrt{2}$. Hence, due to the mentioned construction of the rectangular boxes, each input rectangle can intersect at most two consecutive boxes. Therefore, at most $2n$ such boxes can contain some portion of the rectangles in $R$. Let us denote these boxes by $T_1,T_2,\dots,T_m$, where $m\leq 2n$. Also, the centers of all the portions lie inside these rectangular boxes. Let us denote by $R_i$ ($\subseteq R$) the set of rectangles whose corresponding centers lie in box $T_i$. Now, we have the following result.

\begin{lemma} \label{lx1}
    For each $i\in\{1,2,\dots,m\}$, at most $8$ members of $R$ are needed to dominate all the members of $R_i$. 
\end{lemma}
\begin{proof}
    For a fixed $i$, consider the rectangular box $T_i$. Its side lengths are $\sqrt{2}$ and $2\sqrt{2}$. We divide the box $T_i$ into $8$ square-sized sub-boxes, each having side length $\frac{1}{\sqrt{2}}$ (See Figure \ref{fig:pic11}). Observe that, the distance of every pair of points inside a sub-box is at most $1$. 

    Now, since each input rectangle has its width greater than or equal to 1, the minimum distance between the centers of two non-intersecting portions will be strictly greater than 1. Implying, one sub-box cannot contain the centers of two or more non-intersecting portions. 
    Hence, at most $8$ rectangles will dominate all the elements of $R_i$.
\end{proof}

\begin{figure}[ht]
    \centering
    \includegraphics[width=0.75\linewidth]{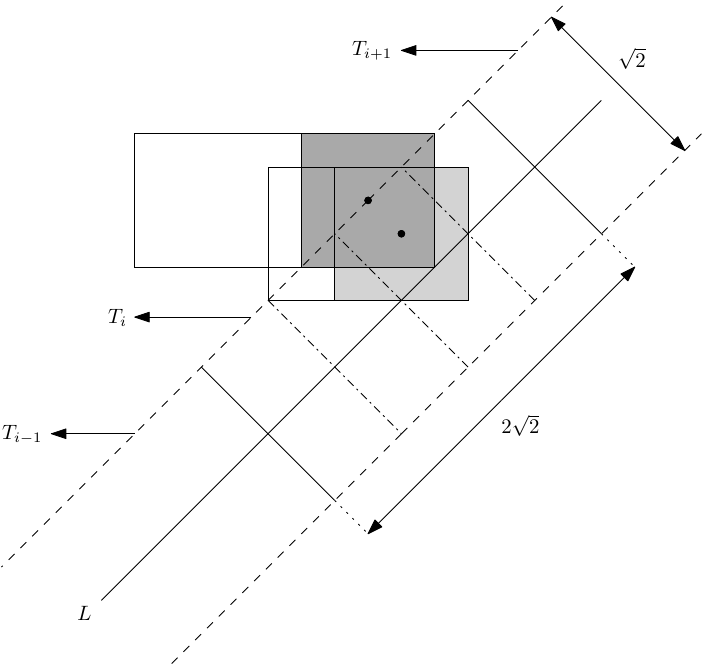}
    \caption{Demonstration of the construction}
    \label{fig:pic11}
\end{figure}

\begin{corollary} \label{corx1}
    At most $24$ members of $R$ are required to dominate all the rectangles intersecting a box $T_i$.
\end{corollary}

\begin{proof}
    All the rectangles intersecting the box $T_i$ is a subset of $R_{i-1}\cup R_{i}\cup R_{i+1}$. Thus, the result follows from Lemma \ref{lx1}. 
\end{proof}

\begin{lemma}
\label{rect}
    There exists a maximum dominating $k$-set, denoted by $OPT$ (say), for which each rectangular box $T_i$ contains the centers of at most $23$ elements of $OPT$. 
\end{lemma}

\begin{proof}
    The result can be proved using Corollary \ref{corx1} and Lemma \ref{lx1}, and following the proof technique of Lemma \ref{LLL} in Section \ref{square}. 
\end{proof}

\textbf{Dynamic programming algorithm:} We need the following notations to design the algorithm for the maximum dominating $k$-set problem for the intersection graph of a set of axis-parallel unit-height rectangles intersected by the straight line $L : y=x$.

\begin{itemize}
    \item $R_i^l$: a subset of $R_i$ of size $l$.
    \item $d(R_i^l)$: set of all rectangles dominated by $R_i^l$.
    \item $d_c(R_i^{l_1}\cup R_{i+1}^{l_2},R_{i+2}^{l_3})$: set of all rectangles dominated by both $R_i^{l_1}\cup R_{i+1}^{l_2}$ and $R_{i+2}^{l_3}$.
\end{itemize}

For a specific choice of the tuple $(l_1,l_2,R_{i-1}^{l_3},R_i^{l_4})$, where $l_1$, $l_2$, $l_3$ and $l_4$ are non-negative integers, we denote by $N(l_1,l_2,R_{i-1}^{l_3},R_i^{l_4})$ the maximum dominated neighborhood size of a $(l_1+l_2+l_3+l_4)$-set, where $R_{i-1}^{l_3}$ is a subset of $l_3$ elements chosen from $R_{i-1}$, $R_i^{l_4}$ is a subset of $l_4$ elements chosen from $R_i$, a subset of $l_2$ elements is chosen from $R_{i-2}$ and a subset of $l_1$ elements is chosen from $\cup_{j=1}^{i-3} R_j$ in that $(l_1+l_2+l_3+l_4)$-set. With these notations, the recursion formula for $N(l_1,l_2,R_{i-1}^{l_3},R_i^{l_4})$ is

$N(l_1,l_2,R_{i-1}^{l_3},R_i^{l_4})$ = 
$\max \{ N(l_1-a,a,R_{i-2}^{l_2},R_{i-1}^{l_3}) + d(R_i^{l_4}) - d_c(R_{i-2}^{l_2}\cup R_{i-1}^{l_3},R_i^{l_4}) :  R_{i-2}^{l_2}\subseteq R_{i-2}$ with $|R_{i-2}^{l_2}|=l_2$, and $a$ is an integer satisfying $0\leq a \leq \min\{l_1,23\}\}$.

By Lemma \ref{rect}, we have $l_2,l_3,l_4\leq 23$ for any choice of the tuple $(l_1,l_2,R_{i-1}^{l_3},R_i^{l_4})$. The choice for $l_1$ can be any integer in $\{0,1,\dots,k\}$.

In the preprocessing step of the dynamic programming algorithm, the values of the functions $d$ and $d_c$ are calculated for all the required subsets of $R$. Next, each box $T_i$ is processed for $i=1,2,\dots,m$, in this order. We use the recursion formula to calculate and store all possible values of $N(l_1,l_2,R_{i-1}^{l_3},R_i^{l_4})$. Finally, the optimal maximum dominated neighborhood size $OPT$ of a $k$-set is obtained as 

$OPT = \max \{ N(l_1,l_2,R_{m-1}^{l_3},R_{m}^{l_4}): R_{m-1}^{l_3}\subseteq R_{m-1}$ with $|R_{m-1}^{l_3}|=l_3$, $R_m^{l_4}\subseteq R_m$ with $|R_m^{l_4}|=l_4$ and $l_1+l_2+l_3+l_4=k, 0\leq l_2, l_3, l_4\leq 23 \}$.

The initialization of the dynamic programming and proof of correctness of the recursion formula are similar to those in the case of unit squares (see  Section \ref{square}).

\textbf{Time complexity:} Calculating $N(l_1,l_2,R_{i-1}^{l_3},R_i^{l_4})$ for a given $(l_1,l_2,R_{i-1}^{l_3},R_i^{l_4})$ requires $O(n^{23})$ time. For each $i$, number of tuples of the form $(l_1,l_2,R_{i-1}^{l_3},R_i^{l_4})$ is $O(kn^{46})$ and there are at most $2n$ possible values of $i$. Hence, the time needed to calculate the function $N$ for all possible arguments is $O(kn^{70})$. Finally, calculating $OPT$ requires $O(n^{23}.n^{23})=O(n^{46})$ time. Hence, the worst case time complexity of the algorithm is $O(kn^{70})+O(n^{46}) = O(kn^{70})$.

\subsection{Case (ii) $L:y=mx$, where $m=\tan\theta$, $0\degree<\theta<45\degree$}
Here we show that if the straight line $L$ makes the angle $\theta$ with the positive $x$-axis (in anticlockwise direction) and $0\degree<\theta<45\degree$, then the method of finding the width of the required strip $T$ needs to be modified.

\begin{observation}
    If we take a strip of width $\sqrt{2}\sin(45\degree+\theta)$ around the straight line $L$ as in Section \ref{square}, then the domination information inside the strip need not always give the optimal solution.
\end{observation}

\begin{figure}[ht]
    \centering
    \includegraphics[width=0.66\linewidth]{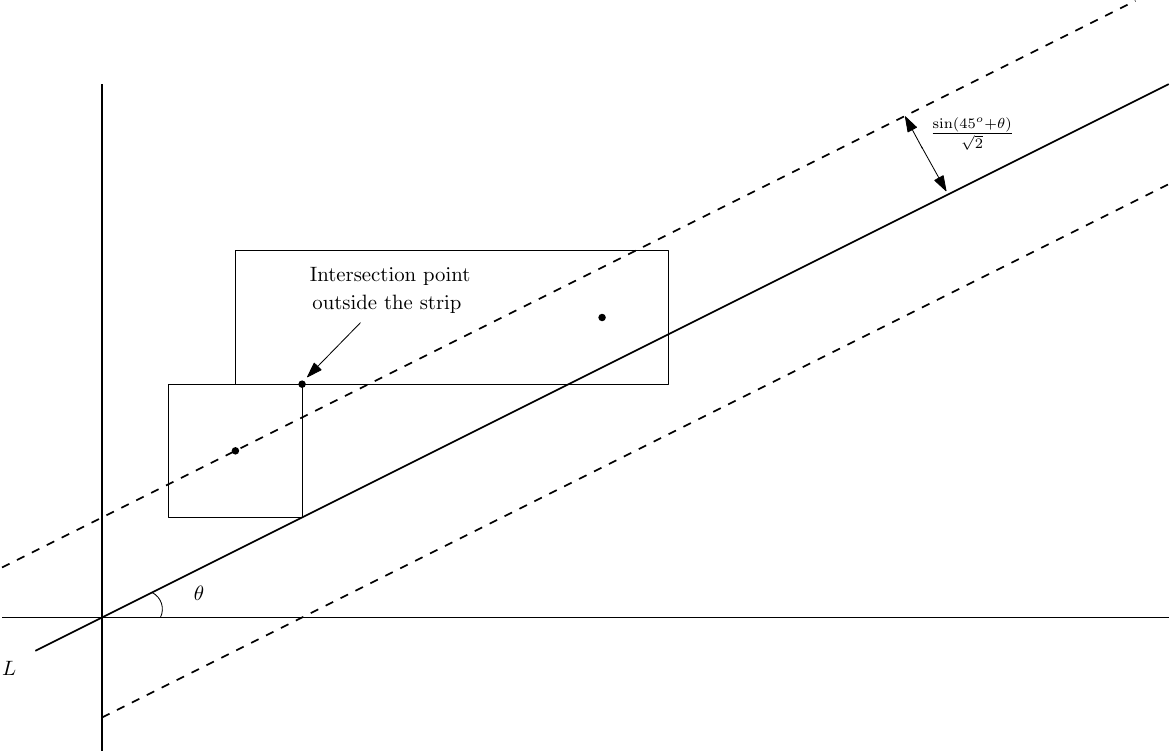}
    \caption{The rectangles have intersection points outside the strip.}
    \label{fig:pic12}
\end{figure}

\begin{proof}
    Consider Figure \ref{fig:pic12}. Here, the rectangles do not intersect inside the strip, but intersect outside. So, if we consider domination information of only inside the strip, the minimum dominating set will have size $2$. But here the optimal minimum dominating set has size $1$. 
\end{proof}

\begin{figure}[ht]
    \centering
    \includegraphics[width=0.9\linewidth]{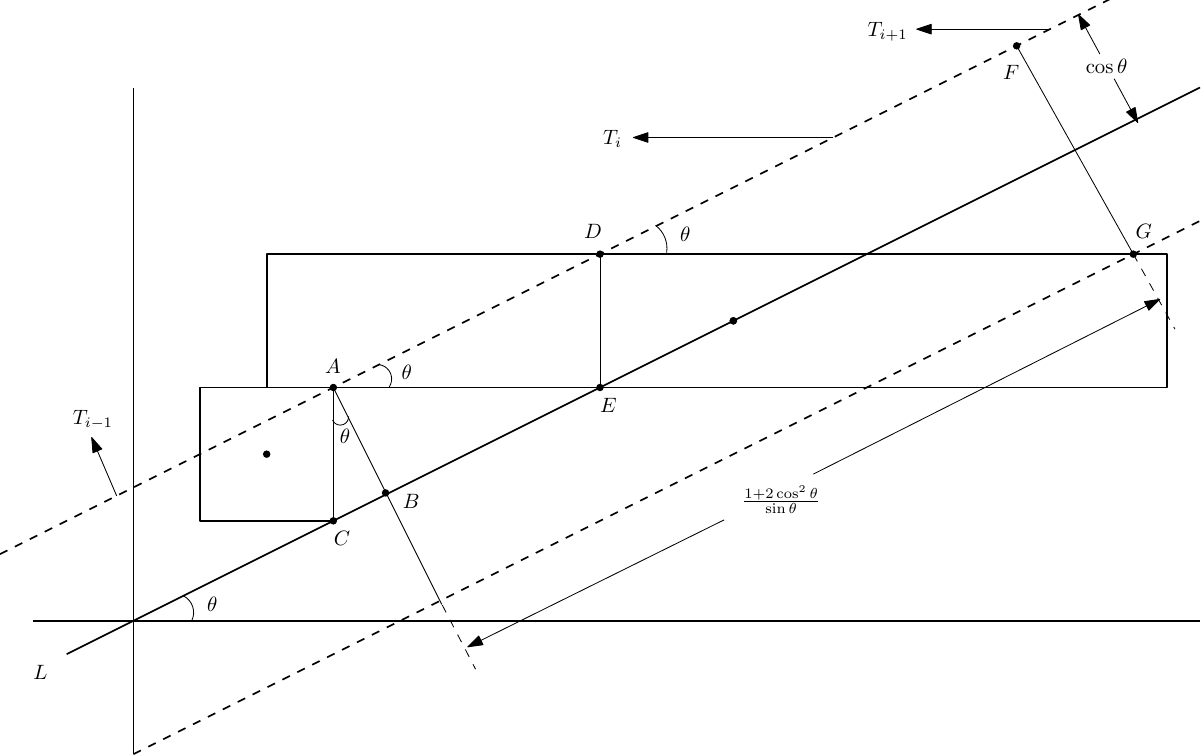}
    \caption{The modified strip $T$ shown above contains all the domination information.}
    \label{fig:pic13}
\end{figure}

To get rid of this problem, let us consider the strip $T$ shown in Figure \ref{fig:pic13}. The width of the strip $T$ is $2AB$.

Now, $\angle BAC=\theta$ and $AC=1$ implies $AB=\cos\theta$. So, this modified strip $T$ has width $2\cos\theta$. 

Also, to ensure that any rectangle can intersect with at most two boxes, we have to take the length of the longer side of each rectangular box $T_i$ to be $AF=AD+DF$.

Since $DE=1$ and $FG=2\cos\theta$, we have $AD=\frac{1}{\sin\theta}$ and $DF=\frac{2\cos\theta}{\tan\theta}$. Hence, $AF=\frac{1+2\cos^2\theta}{\sin\theta}$.

Therefore, the side lengths of each rectangular box $T_i$ will be $2\cos\theta$ and $\frac{1+2\cos^2\theta}{\sin\theta}$ in this case. Assume there are $m$ $(\leq 2n)$ non empty boxes. The following result uses a similar proof technique as Lemma \ref{lx1}.

\begin{lemma}\label{rl1}
    For each $i\in\{1,2,\dots,m\}$, to dominate all the rectangles whose centers lie in $T_i$, at most $\lceil{2\sqrt{2}\cos\theta}\rceil\cdot\lceil{\frac{\sqrt{2}(1+2\cos^2\theta)}{\sin\theta}}\rceil$ ($=\mathcal{P}$, say) number of rectangles of $R$ are required.
\end{lemma}

When the angle $\theta$ satisfies $0\degree<\theta<45\degree$, the quantity $\mathcal{P}$ defined in Lemma \ref{rl1} depends on the given value of $\theta$. An example is shown in Figure \ref{fig:pic14}. Using Lemma \ref{rl1} and applying similar methods as before, we have the following.

\begin{figure}[ht]
    \centering
    \includegraphics[width=0.66\linewidth]{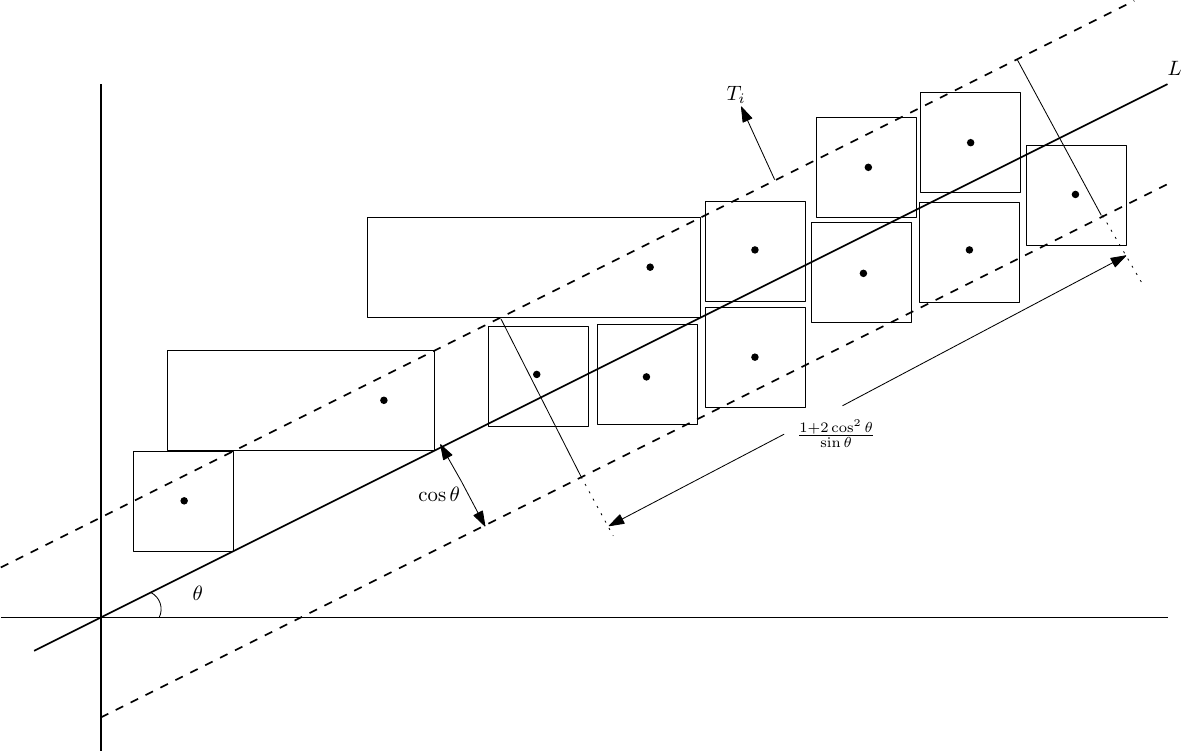}
    \caption{Here, the box $T_i$ contains centers of $10$ non-intersecting rectangles.}
    \label{fig:pic14}
\end{figure}

\begin{theorem}
    The maximum dominating $k$-set problem algorithm for the intersection graph of axis-parallel unit-height rectangles intersected by the straight line $L: y=mx$ with $m=\tan\theta$, $0\degree<\theta<45\degree$ runs in $O(kn^{3(3\mathcal{P}-1)+1})$ time.
\end{theorem}

\begin{remark}
    Observe that when $\theta=0\degree$, the intersection graph of the input rectangles is actually an interval graph where the intervals are the projections of the rectangles on the straight line $L$. Hence, in this case, the maximum dominating $k$-set can be found in $O(n^2k)$ time.
\end{remark}

\begin{remark}
    When $L$ is any other arbitrary straight line, the problem can be solved similarly by considering the angle of $L$ with either the $x$-axis or the $y$-axis.
\end{remark}

\section{Disks intersected by a straight line}\label{di}

\subsection{Case (i): unit disks}
Here, the line $L$ can be any arbitrary straight line in the plane. The construction of strip $T$, the dynamic programming algorithm, the time complexity, and the correctness results are the same as the unit square case of Section \ref{square}. Figure \ref{fig:pic15} illustrates the problem for unit disks.

\begin{figure}[ht]
    \centering
    \includegraphics[width=0.9\linewidth]{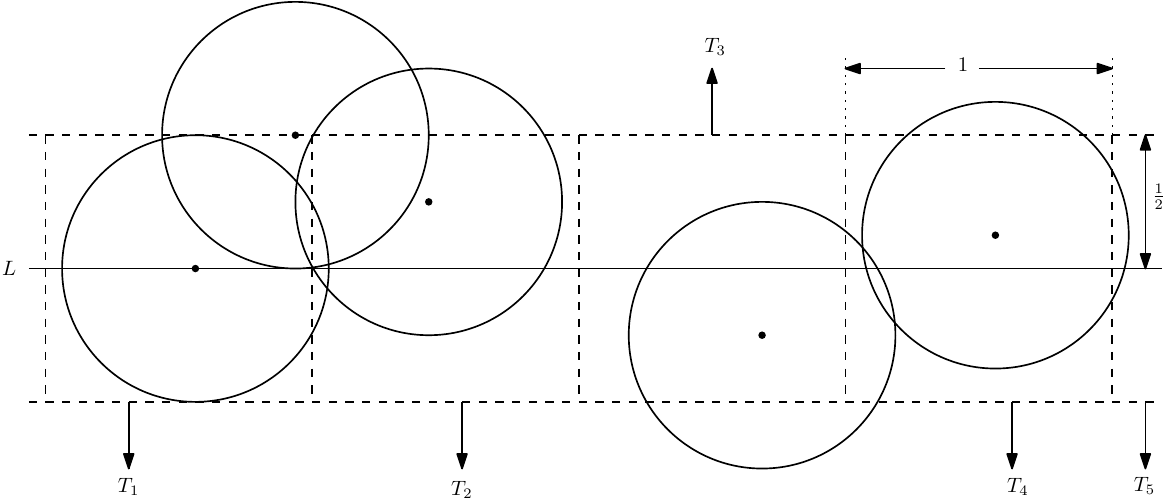}
    \caption{Demonstration of the problem for unit disks}
    \label{fig:pic15}
\end{figure}

\begin{theorem}
    The maximum dominating $k$-set algorithm for the intersection graph of unit disks where a straight line intersects all the disks runs in $O(kn^{34})$ time.
\end{theorem}

\subsection{Case (ii): disks of arbitrary diameter}
Let $n$ input disks in the plane be given, and a straight line intersects all the input disks. We first find the diameters of the largest and the smallest disk; let these be $\mathcal{D}$ and $\mathcal{\delta}$, respectively. 

\begin{figure}[ht]
    \centering
    \includegraphics[width=0.9\linewidth]{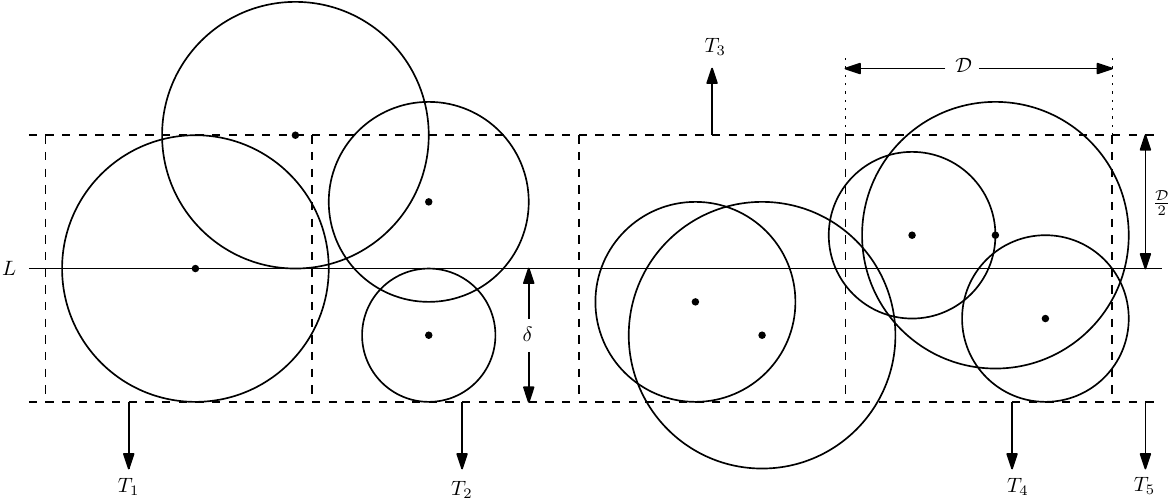}
    \caption{Demonstration of the problem for disks intersected by a straight line}
    \label{fig:pic16}
\end{figure}

In this case, the width of strip $T$ is taken to be $\mathcal{D}$. The strip $T$ is then divided into square-sized boxes $T_i$ of side-length $\mathcal{D}$ (See Figure \ref{fig:pic16}).

Each box $T_i$ is divided into $(\lceil{\frac{\sqrt{2}\mathcal{D}}{\mathcal{\delta}}}\rceil)^2$ sub-boxes, where each sub-box is square-sized and has equal side-length. By construction, each sub-box cannot contain the centers of two or more non-intersecting disks. Hence, by using techniques similar to those in previous sections, we have the following:

\begin{theorem}
    The algorithm for maximum dominating $k$-set problem for the intersection graph of disks where a straight line intersects all the disks runs in $O(kn^{3\{3(\lceil{\frac{\sqrt{2}\mathcal{D}}{\mathcal{\delta}}}\rceil)^2-1\}+1})$ time.
\end{theorem}

\section{Conclusion} \label{conclusion} 
This paper presents some complexity results related to the partial domination problem. Polynomial-time algorithms are proposed for the maximum dominating $k$-set problem for interval graphs and axis-parallel unit square intersection graphs, where a straight line intersects the input unit squares. A polynomial-time algorithm is also given for the intersection graph of unit disks intersected by a straight line. For arbitrary disks intersected by a straight line, a parametrized algorithm is proposed. Improving the time complexities of these algorithms is surely a natural problem. It is also interesting to identify the intersection graph of some other types of geometric objects for which the problem can be solved in polynomial time.

\end{document}